\documentclass[runningheads]{llncs}

\usepackage{hyperref}
\usepackage{enumitem}
\usepackage{amstext}
\usepackage{amssymb}
\usepackage{amsmath}
\usepackage{bm}
\usepackage{relsize}
\usepackage{tikz}
\usetikzlibrary{positioning,arrows}
\usepackage[font={normal,it},labelfont=bf,justification=justified]{caption}
\usepackage{float}
\usepackage{graphicx}
\usepackage[lined,boxed]{algorithm2e}
\usepackage{makeidx}

\newcommand{\OBJ}{{\mathcal{D}}}
\newcommand{\PE}{\mbox{PN}}
\newcommand{\EFF}{\mbox{EFF}}
\newcommand{\LOC}{\mbox{LOC}}

\newcommand{\RE}{\mathbb{R}}
\newcommand{\REd}{\mathbb{R}^d}
\newcommand{\REp}{\mathbb{R}^d_{+}}
\newcommand{\NEp}{\mathbb{N}^d}
\newcommand{\WPP}{\succsim}
\newcommand{\PP}{\succ}
\newcommand{\WST}{\mbox{WST}}

\newcommand{\PROBA}{\mathbb{P}}
\newcommand{\ESP}{\mathbb{E}}
\newcommand{\VAR}{\text{Var}}
\newcommand{\COV}{\text{Cov}}

\newcommand{\CA}{\mathcal{A}}
\newcommand{\CE}{\mathcal{E}}
\newcommand{\CF}{\mathcal{F}}
\newcommand{\CC}{\mathcal{C}}

\newcommand{\NBH}{\mathcal{N}}

\binoppenalty=\maxdimen
\relpenalty=\maxdimen

\begin{document}

\title{On Existence, Mixtures, Computation and Efficiency in Multi-objective Games}

\titlerunning{Multi-objective Games}

\author{Anisse Ismaili}

\authorrunning{Anisse Ismaili}

\institute{RIKEN, Center for Advanced Intelligence Project AIP, Tokyo, Japan\\
anisse.ismaili@riken.jp}

\maketitle            

\begin{abstract}
In a multi-objective game, each individual's payoff is a \emph{vector-valued} function of everyone's actions. Under such vectorial payoffs, Pareto-efficiency is used to formulate each individual's best-response condition, inducing Pareto-Nash equilibria as the fundamental solution concept. In this work, we follow a classical game-theoretic agenda to study equilibria. Firstly, we show in several ways that numerous pure-strategy Pareto-Nash equilibria exist. Secondly, we propose a more consistent extension to mixed-strategy equilibria. Thirdly, we introduce a measurement of the efficiency of multiple objectives games, which purpose is to keep the information on each objective: the multi-objective coordination ratio. Finally, we provide algorithms that compute Pareto-Nash equilibria and that compute or approximate the multi-objective coordination ratio.
\end{abstract}

\vspace*{-5mm}
\section{Introduction}
\vspace*{-2mm}

Game theory and microeconomics assume that individuals evaluate outcomes into scalars. However, bounded rationality can hardly be modeled consistently by agents simply comparing scalars: \emph{``The classical theory does not tolerate the incomparability of oranges and apples.''} \cite{simon1955behavioral}. 
Money is another case of scalarization of the values of outcomes. For instance, while `making money' theoretically creates value \cite{adam1776inquiry}, the tobacco industry making money and killing approximately six million people every year \cite{world2011report} is hardly a creation of value\footnote{Tobacco consumers are free to value and choose cigarettes how it pleases them. However, is value the same when they inhale, as when they die suffocating?}. 

In this work, we assume that agents evaluate outcomes over a finite set of distinct objectives\footnote{It is a backtrack from the subjective theory of value, which typically aggregates values on each objective/commodity into a single scalar by using an utility function.}; hence, agents have vectorial payoffs. For instance, in the case of tobacco consumers, this slightly more informative model would keep the information on these three objectives \cite{conover2014smoking}: smoking pleasure, cigarette cost and consequences on life expectancy. 
In literature, this model was called games with vectorial payoffs, multi-objective games or multi-criteria games; and several applications were considered (see e.g. \cite{zeleny1975games,wierzbicki1995multiple}).
Indeed, behaviors are less assumptively modeled by a partial preference: the Pareto-dominance. Using Pareto-efficiency in place of best-response condition induces Pareto-Nash (PN) equilibria as the solution concept for stability, without even assuming that individuals combine the objectives in a precise manner. Pareto-Nash equilibria encompass the outcomes, even under unknown, uncertain or inconsistent preferences. 

This paper more particularly addresses two unexplored issues.\linebreak
(1) The algorithmic aspects of multi-objective games have never been studied.
(2) Also, the efficiency of Pareto-Nash equilibria has never been a concern. 

\textit{Related literature on mixed-strategies and similar strategy spaces.}\quad 
Games with vectorial payoffs, or multi-objective games, were firstly introduced in the late fifties by Blackwell and Shapley \cite{blackwell1956analog,shapley1959equilibrium}. 
The former shows the existence of a mixed-strategy Pareto-Nash equilibrium in finite two-player zero-sum multi-objective games. 
The later generalizes this existence result to finite multi-objective games. 
Both use a definition of mixed-strategy Pareto-Nash equilibria that suffers an inconsistency: pure-strategy Pareto-Nash equilibria are not included in the set of mixed-strategy Nash equilibria (see Sec. \ref{sec:mixed}). Nonetheless, there is an established literature on games with vector payoffs that uses this definition. 
Deep formal works generalized known existence results \cite{shapley1959equilibrium} to individual action-sets being compact convex subsets of a normed space \cite{wang1993existence}.
Weak Pareto-Nash equilibria can be approximated \cite{morgan2005approximations}. 

\textit{Works related to pure strategies and algorithms.}\quad 
 \cite{wierzbicki1995multiple} achieves to characterize the entire set of Pareto-Nash equilibria by mean of augmented Tchebycheff norms. However, the number of dimensions that parameterize these Tchebycheff norms is algorithmically prohibitive.
\cite{patrone2007multicriteria} shows that a MO potential function guarantees that a Pareto-Nash equilibrium exists in finite MO games.

In Section 3, we show in three different settings that pure-strategy Pareto-Nash equilibria are guaranteed to exist, or very likely to be numerous.
In Section 4, we show an inconsistency in the current concept of mixed-strategy PN equilibrium, and propose an extension to solve this flaw. 
In Section 5, in the fashion of the price of anarchy \cite{koutsoupias1999worst}, we define a measurement of the worst-case efficiency of individualistic behaviors in games, compared to the optimum. In the multi-objective case, it is far from trivial, as  worst-case equilibria and  optima are not uniquely defined. 
In Section 6, we show how to compute the set of (worst) pure-strategy Pareto-Nash equilibria for several game structures, and provide algorithms to compute and approximate our multi-objective coordination ratio.\footnote{For the proofs, see the long paper:}


\section{Preliminaries}

\begin{definition}
A \emph{multi-objective game} (MO game, or MOG) is defined
by the following tuple $\left(N,\{A^{i}\}_{i\in N},\OBJ,\{\bm{u}^{i}\}_{i\in N}\right)$:
\begin{itemize}
\item The agents set is $N=\lbrace1,\ldots,n\rbrace$.
Agent $i$ decides action $a^{i}$ in action-set $A^{i}$.
\item The shared list of objectives is denoted by $\OBJ=\lbrace1,\ldots,d\rbrace$
and every agent $i\in N$ gets her payoff from function $\bm{u}^{i}:A=A^{1}\times\ldots\times A^{n}\rightarrow\mathbb{R}^{d}$
which maps every overall action to a \emph{vector-valued} payoff;
e.g., real $u_{k}^{i}(\bm{a})$ is the payoff of agent $i$ on objective
$k$ for action-profile $\bm{a}=(a^{1},\ldots,a^{n})$.
\end{itemize}
\index{multi-objective game, MO game, MOG}
\index{notations!$N$, set of agents}
\index{notations!$n$, number of agents}
\index{notations!$i$, individual agent}
\index{notations!$A^i$, action-set of agent $i$}
\index{notations!$a^i$, action of agent $i$}
\index{notations!$A=A^{1}\times\ldots\times A^{n}$, set of action-profiles}
\index{notations!$\bm{a}=(a^{1},\ldots,a^{n})\in A$, action-profile}
\index{notations!$\OBJ$, set of objectives}
\index{notations!$k$, an objective}
\index{notations!$d$, number of objectives}
\index{notations!$\bm{u}^{i}:A\rightarrow\mathbb{R}^{d}$, vector-valued payoff function of agent $i$}
\index{notations!$u_{k}^{i}(\bm{a})$, payoff of agent $i$ on objective
$k$ for action-profile $\bm{a}$}
\end{definition}
\begin{figure*}[t]
\caption{Didactic toy example in Ocean Shores city.}
\includegraphics[scale=0.22]{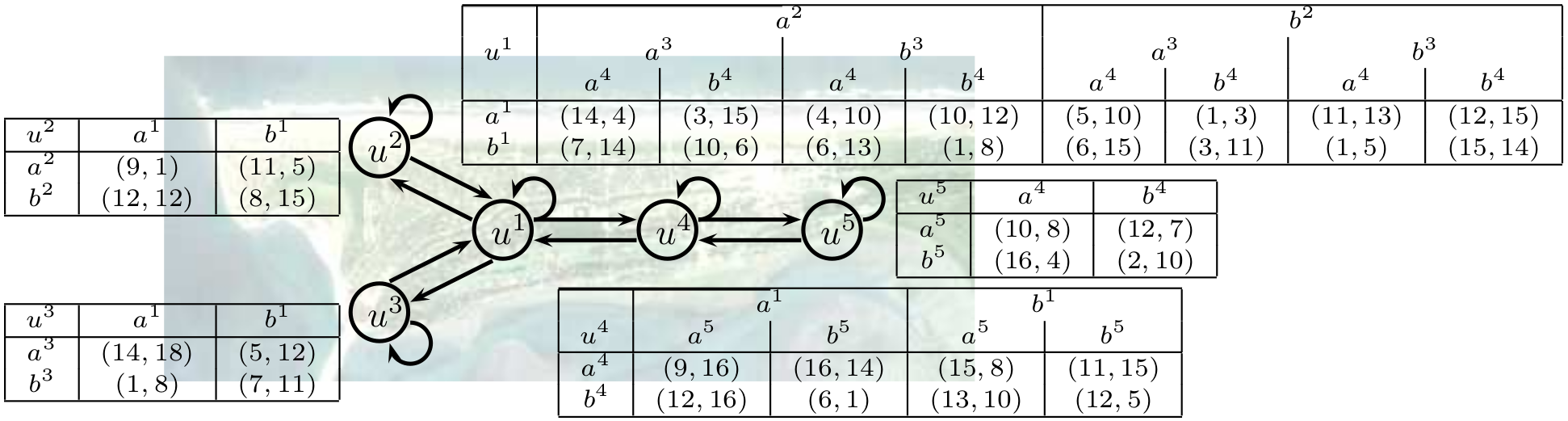}

\small
There are five shops (the nodes) in Ocean Shores: $N=\{1,\ldots,5\}$.  
Each shop/agent $i$ decides between two activities: $A^{i}=\{a^{i},b^{i}\}$;
for instance: renting bikes or buggies, selling clams or fruit, etc. 
That is, agent $i$, in his payoff table $u^i$, decides row $a^{i}$ or $b^{i}$. 
The edges define neighborhoods around every agent.
The payoff of each agent also depends on the actions of her
neighbors, and is differentiated on two objectives $\OBJ=\{1,2\}$ that it would hardly make sence to aggregate, for instance:
sales revenue (to buy their daily lives) and the remaining natural resources (so that, in the future, their children could also live). (Here, the payoffs
are random integers.)
\end{figure*}

In the subjective theory of value, every individual evaluates
her endowment $(u_{1}^{i},\ldots,u_{d}^{i})$ however she wants based
on an utility function $v^{i}:\REd\rightarrow\RE$. The theory of multi-objective
games \cite{blackwell1956analog,shapley1959equilibrium} aims at allowing for individuals that behave according to several unknown, uncertain, or inconsistent utility functions.
These utility functions are reduced to their
common denominator: the Pareto-dominance, as defined below. 
That vector $\bm{y}\in\REd$
weakly-Pareto-dominates and respectively \emph{Pareto-dominates} vector
$\bm{x}\in\REd$ is denoted and defined by:
\index{Pareto-dominance}
\index{weak Pareto-dominance}
\index{notations!$\WPP$, weak Pareto-dominance}
\index{notations!$\PP$, Pareto-dominance}
\begin{eqnarray*} 
\bm{y}\WPP \bm{x} & \Leftrightarrow & \forall k\in\OBJ,\quad  y_k\geq x_k,\\ 
\bm{y}\PP \bm{x} &\Leftrightarrow & \forall k\in\OBJ,\quad y_k\geq x_k\mbox{~~and~~}\exists k\in\OBJ, y_k> x_k.
\end{eqnarray*} 
For the preferences of individuals, given an adversary action-profile\linebreak $\bm{a}^{-i}=(a^j\mid j\neq i)$,
\index{adversary action-profile}
\index{notations!$\bm{a}^{-i}$, adversary action-profile}
this defines a partial rationality
on set\linebreak $\bm{u}^{i}(A^{i},\bm{a}^{-i})=\{\bm{u}^{i}(b^{i},\bm{a}^{-i})\mid b^i\in A^i\}$, which is less assumptive than complete orders,
since it does not presume any individual utility function $v^{i}:\RE^{d}\rightarrow\RE$.
Formally, given a finite set of vectors $X\subseteq\REd$, the set of
\emph{Pareto-efficient} vectors is defined as the following set of
non-Pareto-dominated vectors: \vspace*{-1mm}
\index{Pareto-efficient}
\index{notations!$\EFF[X]$, Pareto-efficient vectors of set $X$}
\index{worst vectors}
\index{notations!$\WST[X]$, worst vectors of set $X$}
\[
\EFF[X]=\{\bm{y}\in X~~|~~\forall \bm{x}\in X,\mbox{~not~}(\bm{x}\PP \bm{y})\}.
\]
Since Pareto-dominance is a partial order,
it induces a multiplicity of Pareto-efficient vectors. These
are the best compromises between objectives. 
Similarly, let $\WST[X]=\{\bm{y} \in X|\forall \bm{x}\in X,\mbox{not}(\bm{y}\PP \bm{x})\}$
denote the worst vectors. 

In a multi-objective
game, individuals behave according to the Pareto -dominance, inducing the solution
concept \emph{Pareto-Nash equilibrium} ($\PE$),\index{Pareto-Nash equilibrium, PN equilibrium}\index{notations!$\PE$, set of Pareto-Nash equilibria} 
formally
defined as any action-profile $\bm{a}\in A$ such that for every agent
$i\in N$: \vspace*{-2mm}
\[
\bm{u}^{i}(a^{i},\bm{a}^{-i})\quad\in\quad\EFF\left[\quad \{\bm{u}^{i}(b^{i},\bm{a}^{-i})\mid b^{i}\in A^{i}\}\quad\right].
\]
We call these conditions \emph{Pareto-efficient responses}.\index{Pareto-efficient responses}
Let $\PE\subseteq A$ denote the set of Pareto-Nash equilibria. For instance, in Figure 1, action-profile $(b^{1},b^{2},a^{3},b^{4},b^{5})$
is a PN equilibrium, since each action, given the adversary local action
profile (column), is Pareto-efficient among the given agent's two actions (rows). In this example, there are $13$ Pareto-Nash equilibria (depicted in Figure 2).

Such an encompassing solution concept provides the first phase for
bounding the efficiency of games. It is well-known that individualistic behaviors can be far from the optimum/maximum in terms of utilitarian evaluation $u(\bm{a})=\sum_{i\in N}u^{i}(\bm{a})$. 
In single-objective games\footnote{In the single-objective case, Pareto-Nash and Nash equilibria coincide.}, this inefficiency
is measured by the \emph{Coordination Ratio} $\mbox{CR}=\frac{\min[u(PN)]}{\max[u(A)]}$
\cite{koutsoupias1999worst}, which is more commonly known as the
\emph{Price of Anarchy} \cite{roughgarden2009intrinsic}. 
However,
in the multi-objective case, the utilitarian social welfare 
\index{utilitarian social welfare}
\index{notations!$\bm{u}:A\rightarrow\RE^d$, utilitarian social welfare}
$\bm{u}(\bm{a})=\sum_{i\in N}\bm{u}^{i}(\bm{a})$
is a vector-valued function $\bm{u}:A\rightarrow\RE^d$ with respect to $d$ objectives.
To study the efficiency of Pareto-Nash equilibria, we introduce:\vspace*{-1mm}
\begin{itemize}
\setlength{\itemsep}{0em}
\item set of \emph{equilibria outcomes} $\quad\CE\quad=\quad\bm{u}(\PE)\quad(\subset\REd),$ 
\item set of \emph{efficient outcomes} $\quad\CF\quad=\quad\EFF[\bm{u}(A)]\quad(\subset\REd)$.
\end{itemize}
\index{equilibria outcomes}
\index{efficient outcomes}
\index{notations!$\CE$, equilibria outcomes}
\index{notations!$\CF$, efficient outcomes}

\begin{figure}[t]
\caption{Biobjective set of utilitarian outcomes $\bm{u}(A)\subset\RE^2$ in Ocean Shores.}
\begin{minipage}{0.4\textwidth}
\includegraphics[scale=0.33]{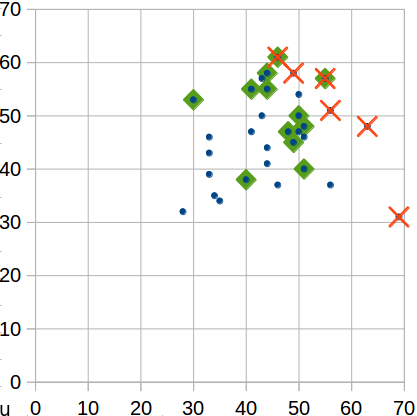}
\end{minipage}
~
\begin{minipage}{0.57\textwidth}
\small
The utilitarian outcomes are a set of vectors, depicted above. 
Worst case equilibria and optima are not uniquely defined.
The ratio of set of equilibria outcomes
$\CE$ ($\Diamond$) to set of efficient outcomes $\CF$ ($\times$) would be a ratio of sets, which remains undefined. It would be crucial that such a definition keeps information for every objective. E.g., we want to remember that a car pollutes, or that a cigarette kills, not just that it makes some economic agents happy.
\end{minipage}
\end{figure}

\vspace*{-4mm}
\section{Numerous pure strategy Pareto-Nash equilibria exist.}
\vspace*{-2mm}

This section demonstrates the existence of pure strategy Pareto-Nash equilibria. 
Firstly, 
we write how the existence results from single-objective (SO) games can be retrieved in MO games.\index{single-objective, SO}
Secondly,
we generalize the equilibria existence results of single-objective potential games to multi-objective potential games.
Thirdly, we show that on average, numerous Pareto-Nash equilibria exist.

\subsection{Reductions from MO games to SO games}

In the literature, most rationalities are constructed by means of a utility function
$v^{i}:\RE^{d}\rightarrow\RE$, which is monotonic with respect to
the Pareto-dominance, that is:\vspace*{-2mm}
\begin{eqnarray*}
\bm{x}\PP \bm{y} &\Rightarrow& v^{i}(\bm{x})> v^{i}(\bm{y})
\end{eqnarray*}
Such functions are called \emph{Pareto-monotonic}\index{Pareto-monotonic function}.
For instance, these include
positive weighted sums, Cobb-Douglas utilities, and utility functions in general
as assumed by the Arrow-Debreu theorem. 

A straightforward consequence is that the set of Pareto-efficient
vectors contains the optima of any Pareto-monotonic utility function. Formally, given
a MOG $\Gamma$, from Pareto-monotonic utility functions
$V=(v^{i}:\REd\rightarrow\RE|i\in N)$ the single-objective game $V\circ\Gamma=(N,\{A^{i}\}_{i\in N},\{v^{i}\circ \bm{u}^{i}\}_{i\in N})$
results from the given utilities, and one has: $\PE(V\circ\Gamma)\subseteq\PE(\Gamma).$
\index{notations!$V\circ\Gamma$, composition of utilities and an MOG into a single-objective game}
In other words, Pareto-Nash equilibria
encompass the game's outcome, regardless of the unknown
preferences.

Also, inclusion $\PE(V\circ\Gamma)\subseteq\PE(\Gamma)$
 argues for the guaranteed existence of numerous PN equilibria in MO games, under the following assumptions:\vspace*{-1mm}
\begin{enumerate}
\setlength{\itemsep}{0em}
\item the structure of the SO game on every objective is the same,
\item equilibria are guaranteed in that structure of SO game, 
\item and a positive linear combination of the MO game induces that SO game.
\end{enumerate} 
This remark is the canonical argument used in previous
results 
(e.g. \cite{shapley1959equilibrium,patrone2007multicriteria}).

\subsection{Multi-objective potentials}

We now explore potential games, as introduced for congestion games
by Robert Rosenthal \cite{rosenthal1973class,monderer1996potential}
and recently generalized to MO games \cite{patrone2007multicriteria}.
The existence of an MO potential function guarantees that at least
one Pareto-Nash equilibrium exists \cite{patrone2007multicriteria}. 
We go further and completely characterize the set of PN equilibria.

\begin{definition}
An MO game $\Gamma=\left(N,\{A^{i}\}_{i\in N},\OBJ,\{\bm{u}^{i}\}_{i\in N}\right)$
admits \emph{(exact) potential function} $\bm{\Phi}:A\rightarrow\RE^{d}$
if and only if 
for every action-profile $\bm{a}\in A$, for every agent
$i\in N$ and for every action $b^{i}\in A^{i}$, one has: 
\index{potential function}
\[
\forall k\in\OBJ,\quad
\Phi_k(b^{i},\bm{a}^{-i})-\Phi_k(\bm{a})\quad=\quad u^{i}_k(b^{i},\bm{a}^{-i})-u^{i}_k(\bm{a}).
\]
\end{definition}
That is, function $\bm{\Phi}$ additively accumulates the vectorial values of each deviation. 

\begin{definition}
Given a vector valued function $\bm{\Phi}:A\rightarrow\RE^{d}$, let the set of \emph{locally efficient} action-profiles $\LOC(\bm{\Phi})$ be the set of action-profiles $\bm{a}\in A$ such
that:
\index{locally efficient}
\index{notations!$\LOC(\bm{\Phi})$, locally efficient action-profiles}
\[
\bm{\Phi}(\bm{a})\quad\in\quad\EFF[\{\bm{\Phi}(b^{i},\bm{a}^{-i})\in\RE^{d}\mid i\in N,b^{i}\in A^{i}\}].
\]
\end{definition}
Set $\LOC(\bm{\Phi})$ corresponds to a generalization of local optima for function $\bm{\Phi}$, and is non-empty if sets $N$, $\OBJ$ and $A$ are finite. Moreover, due to the loose requirement for local efficiency, set $\LOC(\bm{\Phi})$ is likely to contain numerous action-profiles.

\begin{theorem}\label{th:1}
Let $\Gamma=\left(N,\{A^{i}\}_{i\in N},\OBJ,\{\bm{u}^{i}\}_{i\in N}\right)$
be a \emph{finite multi-objective game}\index{finite multi-objective game}\footnote{In a finite multi-objective game, sets $N$, $\{A^{i}\}_{i\in N}$  and $\OBJ$ are finite.} that admits potential function $\bm{\Phi}$.
Then, it holds that:
\[
\PE(\Gamma)\quad=\quad\LOC(\bm{\Phi})\quad\neq\quad\emptyset.
\]
\end{theorem}
This theorem completely characterizes
the set of Pareto-Nash equilibria as the set of locally efficient
action-profiles for function $\bm{\Phi}$, which is a non-empty set with numerous action-profiles. More generally, Theorem \ref{th:1} also holds when sets $N$ and $\OBJ$ are finite and sets $A^i$ are just compact.

\subsection{Likelihood of equilibrium in random games}

Another manner to study whether a $\PE$-equilibrium exists is to provide
a probability distribution on a family of finite games and then discuss the probability
of $\PE$-equilibrium existence. A similar methodology was successfully
applied \cite{goldberg1968probability,dresher1970probability,rinott2000number}
to SO games in several settings where every SO payoff $u^{i}(\bm{a})$
is independently and identically distributed by a uniform distribution
on continuous intervals $[0,1]$. At the heart of this subsection,
let random variable $Z$ denote the number of pure Nash-equilibria action-profiles
in the game. In the SO case, there is almost surely only one
best response. However, when considering MO games, a main technical difference
lies in the average number of ``best responses'' (or here, Pareto-efficient
responses), which in most cases exceeds $1$, due to the surface-like shape of the Pareto-efficient set in $\RE^{d}$, surface which is $(d-1)$
dimensional.
Here, we assume a probability distribution $\PROBA_{n,\alpha,\beta}$,
that builds randomly the Pareto-efficient response tables of an $n$-agent normal form game with $\alpha$
actions-per-agent: \index{notations!$\alpha$, number of actions-per-agent}  for every agent
$i$ and every adversary action-profile $\bm{a}^{-i}\in\prod_{j\neq i}A^{j}$,
there is a fixed number $\beta:1<\beta\leq\alpha$ of Pareto-efficient
responses, for the sake of simplicity.
 \begin{theorem} 
 Given numbers $n\geq2$ of agents, $\alpha\geq2$
of actions-per-agent and $\beta\leq\alpha$ of Pareto-efficient responses,
based on probability distribution $\PROBA_{n,\alpha,\beta}$, the
number $Z$ of Pareto-Nash equilibria satisfies $\ESP[Z] = \beta^n$ and:
\begin{eqnarray*} 
\PROBA\left( (1-\gamma)\beta^n\leq {Z} \leq (1+\gamma)\beta^n \right) & \geq & 1 - \frac{1}{\gamma^2 \beta^{n}},\quad\forall\gamma\in(0,1). \end{eqnarray*} \end{theorem} 
It argues for the existence of numerous Pareto-Nash equilibria when
there are enough agents and efficient responses, and follows from the Bienaym\'e-Tchebychev inequality.
For instance, (given
$\gamma=1/2$) the probability that the number of Pareto-Nash equilibria $Z$ is between $(1/2)\beta^{n}$ and $(3/2)\beta^{n}$, is at least $1-4\beta^{-n}$, which
for $\beta=2$ efficient responses and $n=5$ agents, gives $\PROBA(16\leq Z\leq48)\geq7/8$.

\section{Consistent extension to mixed strategies}\label{sec:mixed}

To guarantee equilibrium existence by means of fixed-point theorems on compact sets \cite{vnmAndMorgenstern1944,nash1950equilibrium}, the finite action sets of every agent are expanded to include \emph{mixed strategies}. That is:
every agent $i$ decides a probability distribution
$p^{i}$ in the set $\Delta(A^{i})$ of probability distributions over
his action-set $A^{i}$. 
\index{mixed-strategy}
\index{notations!$p^{i}$ mixed strategy of agent $i$}
\index{notations!$\Delta(A^{i})$ set of mixed strategies of agent $i$}
Each payoff function $\bm{u}^{i}$ is redefined to be the expected utility
\begin{eqnarray*}
\bm{u}^{i}(\bm{p}) &=& \ESP_{\bm{a}\sim\bm{p}}[\bm{u}^{i}(\bm{a})],
\end{eqnarray*}
under the mixed-strategy profile $\bm{p}=(p^{1},\ldots,p^{n})\in\prod_{i\in N}\Delta(A^{i})$.
This defines a mixed-extension\index{mixed extension} of the original game.
The stability concept induced is called a mixed-strategy Nash equilibrium.

In MOGs, Pareto-Nash equilibria based on
their original definition by Blackwell \cite{blackwell1956analog} and Shapley \cite{shapley1959equilibrium} (below) are those usually considered \cite{borm1988pareto,corley1985games,voorneveld1999potential,zeleny1975games}.

\begin{definition}
\index{mixed-strategy Pareto-Nash equilibrium}
Given
finite MO game $\Gamma=\left(N,\{A^{i}\}_{i\in N},\{\OBJ\},\{\bm{u}^{i}\}_{i\in N}\right)$,
a mixed-strategy profile  $\bm{p}=(p^{1},\ldots,p^{n})\in\prod_{i\in N}\Delta(A^{i})$
is a mixed-strategy Pareto-Nash equilibrium if and only if it satisfies for every agent $i$: 
\[
\bm{u}^{i}(p^{i},\bm{p}^{-i})\in\EFF\left[\left\{ \bm{u}^{i}(q^{i},\bm{p}^{-i})\in\RE^{d}\mid q^{i}\in\Delta(A^{i})\right\} \right]
\]
\end{definition}

\quad The rational behind this first definition is the following.
For every agent $i$, mixed-strategy
$p^{i}\in\Delta(A^{i})$ acts as a convex-combination of set of vectorial
payoffs $\bm{u}^{i}(A^{i},\bm{p}^{-i})$ and the best-response
condition is replaced by the fact that mixed-strategy $p^{i}$ should
have a Pareto-efficient evaluation $\bm{u}^{i}(p^{i},\bm{p}^{-i})$ among the elements of this convex
set of evaluations $\{\bm{u}^{i}(q^{i},\bm{p}^{-i})\in\RE^{d}\mid q^{i}\in\Delta(A^{i})\}$.
That is, a mixed-strategy Pareto-Nash equilibrium is a pure-strategy
Pareto-Nash equilibrium in finite game $\Gamma$'s mixed extension.
However, as depicted in Figure \ref{fig:mix}, Definition 1 
fails to fulfill two
fundamental requirements: 
\begin{enumerate}
\item Pure-strategy equilibria must be included
in mixed-strategy equilibria. 
\item Mixed-strategies also enable to model a risk-averse agent.
\end{enumerate}
\begin{proof} Figure 3 demonstrates these side effects.\end{proof}

To fulfill the two requirements, instead of efficient mixed actions, we consider mixtures of efficient pure-actions. As in Figure \ref{fig:mix}, it corrects both side effects.

\begin{definition}
Given a finite multi-objective game $\left(N,\{A^{i}\}_{i\in N},\{\OBJ\},\{\bm{u}^{i}\}_{i\in N}\right)$,
a mixed-strategy Pareto-Nash equilibrium is a mixed-strategy profile\linebreak
$\bm{p}=(p^{1},\ldots,p^{n})\in\prod_{i\in N}\Delta(A^{i})$, such
that for every agent $i$ and action $a^{i}\in A^{i}$ if $a^{i}$
is played with positive probability $p^{i}(a^{i})>0$, then it holds
that 
\[
\bm{u}^{i}(a^{i},\bm{p}^{-i})\quad\in\quad\EFF\left[\bm{u}^{i}(A^{i},\bm{p}^{-i})\right].
\]
\end{definition}

\begin{figure}[t]
\caption{Single-agent three-actions bi-objective game showing inconsistencies. (The coordinates correspond to the bi-objective valuation $(u_1,u_2)$.)}\label{fig:mix}
\begin{minipage}{0.38\textwidth}
\includegraphics[scale=0.16]{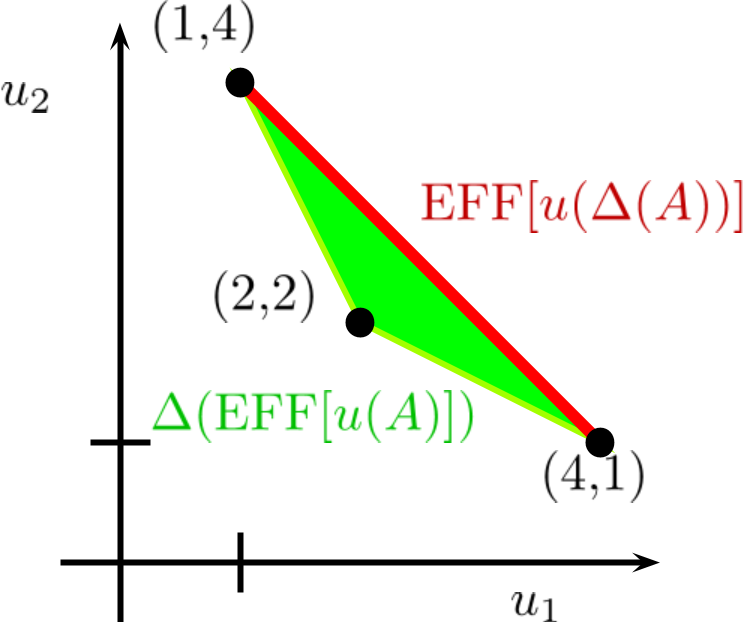}
\end{minipage}
~
\begin{minipage}{0.62\textwidth}
The three outcomes, $u(A)=\{(1,4),(2,2),(4,1)\}$,
are depicted by black dots. With Def. 4, since the
mixed outcomes are all convex-combinations of $\{(1,4),(2,2),(4,1)\}$,
the Pareto-efficient mixed-strategies are here the convex-combinations
of $\{(1,4),(4,1)\}$; and outcome $(2,2)$ is Pareto-dominated. Not every pure-strategy Pareto-Nash equilibrium is a mixed-strategy one, which is a severe inconsistency. Furthermore, since outcome
$(2,2)$ is well balanced, it may also be decided with a non-null probability, e.g.,
if the agent's utility is concave \cite{cobb1928theory}, or if she is risk-averse \cite{kahneman1979prospect}. Our
\emph{revised} definition considers instead all the convex-combinations of the Pareto-efficient pure actions
$\{(1,4),(2,2),(4,1)\}$.
\end{minipage}
\end{figure}

\quad This generalized definition connects in the single-objective case to a less know definition of Nash-equilibria (see \cite{papadimitriou2007agt}, page 30, Theorem 2.1). In this alternative definition, each mixed strategy must be a mixture of pure-strategies that are best-responses. In other words, the support of each mixed strategy must be included in the set of pure-strategy best-responses. 
Furthermore, concerning existence, since this revised definition contains the former one, (which is guaranteed to exist) the new definition is guaranteed to exist too.

\section{Multi-objective coordination ratio }

In the single-objective case, the coordination ratio measures the efficiency loss of equilibria compared to the optimum.
In MO games, we claim that it is critical 
to study efficiency with respect
to every objective.
Even after the actions, the game analyst still has access to the vectorial payoffs.  In this section, we follow
the agenda outlined in the introduction, to define a \emph{multi-objective
coordination ratio} $\text{MO-CR}[\CE,\CF]$ of the set of equilibria outcomes
$\CE$ to the set of efficient outcomes $\CF$, that fills the critical purpose
to keep information on each objective. 

First, we state the list of desirable properties that we want the ratio to satisfy.
For the purpose of having meaningful divisions and ratios, 
some vectors are positive in this section.
Given vectors $\bm{\rho},\bm{y}\in\REd$ and $\bm{z}\in\REp$, vector $\bm{\rho}\star \bm{y}\in\REd$ is defined
by $\forall k\in\OBJ,(\bm{\rho}\star \bm{y})_{k}=\rho_{k}y_{k}$. Vector
$\bm{y}/\bm{z}\in\REd$ is defined by $\forall k\in\OBJ,(\bm{y}/\bm{z})_{k}=y_{k}/z_{k}$.
Given vector $\bm{r}\in\REd$ and set of vectors $Y$, set $\bm{r}\star Y$ is defined by $\{\bm{r}\star \bm{y}\in\REp|\bm{y}\in Y\}$ and
for $\bm{r}\in\REp$, set $Y/\bm{r}$ is defined by $\lbrace \bm{y}/\bm{r}\in\REd|\bm{y}\in Y\rbrace$. 
Given $\bm{x}\in\REd$,
cone $\CC(\bm{x})$ denotes $\{\bm{y}\in\REd~|~\bm{x}\WPP \bm{y}\}$, and given $X\subset\REd$,
cone-union $\CC(X\text{)}$ is defined by $\cup_{\bm{x}\in X}\CC(\bm{x})$.
Vector $\bm{0}$ denotes a vector with $d$ zeros, and $\bm{1}$ denotes a vector with $d$ ones.

\index{notations!$\bm{\rho}\star \bm{y}$, product of vectors into a vector}
\index{notations!$\bm{y}/\bm{z}$, division of vectors into a vector}
\index{notations!$\bm{r}\star Y$, product of a vector and a set of vectors into a set of vectors}
\index{notations!$Y/\bm{r}$, division of a set of vectors by a vector into a set of vectors}
\index{cone}
\index{cone-union}
\index{notations!$\CC(\bm{x})$, cone}
\index{notations!$\CC(X)$, cone-union}

The first property that we require from $\mbox{MO-CR}[\CE,\CF]$ is to be on a \emph{multi-objective ratio
scale}\index{multi-objective ratio scale}. Given $\CE,\CF\subset\REp$ and $\bm{r}\in\REp$, the following shall hold. 
\begin{eqnarray} 
\mbox{MO-CR}[\CE,\CF] 
&\quad\subseteq& 
\REd\label{eq:ratio:6}\\ 
\mbox{MO-CR}[\{\bm{0}\},\CF] 
&\quad=& 
\{\bm{0}\}\label{eq:ratio:7}\\ 
\mbox{MO-CR}[\bm{r}\star\CE,\CF] 
& \quad=& 
\bm{r}\star\mbox{MO-CR}[\CE,\CF]\label{eq:ratio:8}\\ 
\mbox{MO-CR}[\CE,\bm{r}\star\CF] 
& \quad=& 
\mbox{MO-CR}[\CE,\CF]/\bm{r}\label{eq:ratio:9}\\ 
\CE\subseteq\CF 
& \quad\Leftrightarrow& 
\bm{1}\in \mbox{MO-CR}[\CE,\CF]\label{eq:ratio:10} 
\end{eqnarray}
To fix these ideas one can think of $d=1$ and given two positive numbers $e,f$, to the properties of ratio $e/f$.
Equation (\ref{eq:ratio:6}) states that MO-CR is expressed in a multi-objective
space. Equations (\ref{eq:ratio:7}), (\ref{eq:ratio:8}) and (\ref{eq:ratio:9})
state that MO-CR is well-centered and sensitive on each objective
to multiplications of outcomes, which is what we want.
For instance, if $\CE$ is three times
better on objective $k$, then so is MO-CR. If there are two times
more efficient opportunities in $\CF$ on objective $k'$, then MO-CR
is one half on objective $k'$. In other words, the efficiency of
each objective independently reflects on MO-CR in a ratio-scale. Equation
(\ref{eq:ratio:10}) states that if all equilibria outcomes are efficient
(i.e. $\CE\subseteq\CF$), then this amounts to $\bm{1}\in\text{MO-CR}[\CE,\CF]$,
i.e. the MO game is fully efficient.

These requirements rule out a set of first ideas.
For instance, we can rule out comparisons of equilibria outcomes to ideal vector $\mathcal{I}=(\max_{z\in\CF}\lbrace z_{k}\rbrace|k\in\OBJ)$
does not satisfy requirement (5) to have $\bm{1}\in\text{MO-CR}[\CE,\CF]$ when $\CE\subseteq\CF$.
By starting from a social
welfare $f:\REp\rightarrow\RE_+$, taking ratio $\min f(\CE)/\max f(\CF)$,  
induces the same problem. 

This measurement should also be non-dictatorial, in the sense that no point of
view should be imposed on what the overall efficiency is: no
prior choice must be done on the set of efficient outcomes. Formally,
if two sets of efficient outcomes $\CF,\CF'\subset\REp$ differ even
slightly, then this must reflect at least for some numerator set $\CE$ onto ratio $\text{MO-CR}[\CE,\CF]$. 
This amounts to a disjunction on efficient outcomes.
Finally $\text{MO-CR}[\CE,\CF]$
must provide guaranteed efficiency ratios that hold for every equilibrium
outcome $\bm{y}\in\CE$, which amounts to a conjunction on equilibria outcomes.
The definition below follows from these requirements. 

Firstly, the efficiency of \emph{one} equilibrium $\bm{y}\in\CE$ is quantified
without prior choices on what efficient outcome should we compare it to, as required:
\begin{eqnarray*} 
R[\bm{y},\CF]\quad =\quad \bigcup_{\bm{z}\in\CF}\CC(\bm{y}/\bm{z}), 
\end{eqnarray*} 
The idea is that we do not take sides with any efficient outcome. 
Instead, we define with flexibility
and without a dictatorship a disjunctive set of guaranteed
efficiency ratios, 
which lets the differences between two sets of efficient outcomes $\CF,\CF'\subset\REp$ reflect onto ratio $\text{MO-CR}[\CE,\CF]$.

Secondly, in MOGs, on average, there are many Pareto-Nash equilibria.
An efficiency \emph{guarantee} $\bm{\rho}\in\REd$ should hold for
every equilibrium outcome. It induces this conjunctive definition of
the set of guaranteed vectorial ratios: 
\index{guaranteed vectorial ratios}
\index{notations!$R[\CE,\CF]$, set of guaranteed vectorial ratios}
\begin{eqnarray*} 
R[\CE,\CF]\quad=\quad\bigcap_{\bm{y}\in\CE} R[\bm{y},\CF]. 
\end{eqnarray*}

In fact, because of the conjunction on equilibria outcomes, the set $R[\CE,\CF]$ only depends on sets $\WST[\CE]$ (instead of set $\CE$) and $\CF$.

Finally, if two bounds on efficiencies $\bm{\rho}$ and $\bm{\rho}'$ are such
that $\bm{\rho}\PP\bm{\rho}'$ (e.g. the former guarantees fraction $\bm{\rho}=(0.75,0.75)$ of efficiency and the later fraction $\bm{\rho}'=(0.5,0.5)$), then $\bm{\rho}'$ brings no more information; hence,
MO-CR is defined using $\EFF$ on the guaranteed efficiency ratios
$R[\WST[\CE],\CF]$. These points are summed up in the following definition:
\begin{definition}[MO-CR]\label{def:MOCR}
Given an MO game, vector $\bm{\rho}\in\REd$
bounds its inefficiency (i.e. $\bm{\rho}\in R[\CE,\CF]$) if and only if the following holds (see Fig. 4) : 
$$\forall \bm{y}\in\CE,\quad\exists \bm{z}\in\CF,\quad \bm{y}/\bm{z}\WPP\bm{\rho}.$$
The multi-objective coordination ratio $\text{MO-CR}[\CE,\CF]$
is then defined as: 
\index{multi-objective coordination ratio}
\index{notations!$\text{MO-CR}[\CE,\CF]$, MO coordination ratio}
\[
\text{MO-CR}[\CE,\CF]\quad=\quad\EFF[R[\WST[\CE],\CF]].
\]
\end{definition}

\begin{figure}[t]
\caption{Didactic depiction of a guaranteed vectorial ratio $\bm{\rho}$ from $\text{MO-CR}[\CE,\CF]$.}
\begin{minipage}{0.44\textwidth}
\includegraphics[scale=0.75]{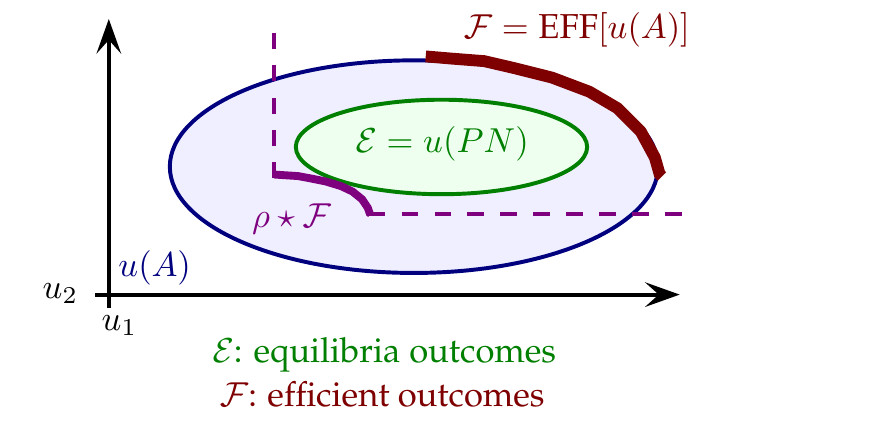}
\end{minipage}
~
\begin{minipage}{0.55\textwidth}
\small
{The multi-objective coordination ratio can be explained by 
the implications of a vectorial ratio $\bm{\rho}\in\text{MO-CR}$: 
for each vector $\bm{y}\in\CE$, an efficient outcome
$\bm{z}^{(\bm{y})}\in\CF$ exists such that $\bm{y}$ Pareto-dominates vector $\bm{\rho}\star \bm{z}^{(\bm{y})}$.
In other words, equilibria outcomes $\CE$ are at least as good as set of vectors
$\bm{\rho}\star\CF$: If $\bm{\rho}\in R[\CE,\CF]$, then every equilibrium
satisfies the ratio of efficiency $\bm{\rho}$ in an unspecified manner.
In other words, the equilibria outcomes are contained in the ``at least as good as
$\bm{\rho}\star\CF$'' cone-union, that is: $\CE\subseteq(\bm{\rho}\star\CF)+\REp.$ Moreover, since $\bm{\rho}$ is tight, set $\CE$ sticks to $\bm{\rho}\star\CF$.}
\end{minipage}
\end{figure}

The most famous results of the coordination ratio (or price of anarchy)
are stated analytically on families of games, for instance on congestion
games \cite{christodoulou2005price,roughgarden2009intrinsic}. Such results
would also be desirable in the multi-objective case. However, the underlying
proofs do not survive this generalization: while best response inequalities
can be summed in single-objective cases, here, non-Pareto-dominances cannot.
This issue is independent of the chosen efficiency measurement and motivates
numerical approaches, as proposed in the next section.

\section{Computation}

In this section, we provide algorithms for computing the set of pure-strategy
Pareto-Nash equilibria and for computing the multi-objective coordination ratio. 

\subsection{Computing pure-strategy Pareto-Nash equilibria}

If the MO game is given in \emph{normal form}\index{normal form}, then it is made of the MO
payoffs of every agent $i\in N$ on every action-profile $\bm{a}\in A$.
Since there are $n\alpha^{n}$ such vectors, where recall that $n$ is the number of agents, $\alpha$ the number of actions per agent and $d$ the number of objectives, the length of this input is $L(n)=n\alpha^{n}d$. Then, enumeration of the action-profiles works
efficiently 
with respect to length function $L$, using a simple argument  similar to \cite{gottlob2005pure}.
\index{length function}
\index{notations!$L$, length function}

\begin{theorem} Given a MO game in normal form, computing the set
of the best (resp. worst) equilibria outcomes $\EFF[\CE]$ (resp.
$\WST[\CE]$) takes polynomial time  
\[
O(n\alpha^{n+1}d+\alpha^{2n}d)\quad=\quad O(L^{2}).
\]
 Moreover, if $d=2$, this complexity is lowered to quasi-linear-time
$$O(n\alpha^{n}\log_{2}(\alpha))=O(L\log_{2}(\alpha)).$$ 
\end{theorem}\medskip{}

\quad {Graphical games} provide compact representations
of massive multi-agent games when the payoff functions of the agents
only depend on a local subset of the agents \cite{kearns2001graphical}. 
Graphical games can be generalized in a straightforward manner 
to assuming vectorial payoffs. Formally, there is a support graph $G=(N,E)$ where each vertex represents an agent, and an agent $i$'s evaluation function only depends on the actions of the agents in his inner-neighbourhood $\NBH(i)=\{j\in N|(j,i)\in E\}$. That is $\bm{u}^i:A^{\NBH(i)}\rightarrow\RE^d$ maps each local action-profile $\bm{a}^{\NBH(i)}\in A^{\NBH(i)}$ to a multi-objective payoff $\bm{u}^i(\bm{a}^{\NBH(i)})\in\RE^d$.\index{notations!$\NBH(i)$, neighborhood of vertex $i$}
\begin{definition}[Multi-objective graphical game (MOGG)]\index{multi-objective graphical game (MOGG)}
An MOGG is a tuple $\left(G=(N,E), \{A^i\}_{i\in N}, \OBJ, \{\bm{u}^i\}_{i\in N}\right)$. $N$ is the set of agents. $\{A^i\}_{i\in N}$ are their individual action-sets. $\OBJ$ is the set of all objectives. Every function $\bm{u}^i:A^{\NBH(i)}\rightarrow\RE^d$ is vector-valued, and its scope is vertex $i$'s neighborhood. 
\end{definition}
Figure 1 pictures a didactic instance of an MOGG.
In the same manner as
computing equilibria in graphical games was reduced to junction-tree
algorithms \cite{daskalakis2006computing}, it is also possible to
exploit a generalized MO junction-tree algorithm
\cite{dubus2009multiobjective,gonzales2011decision}. However, even
though this MO junction-tree algorithm is not in polynomial
time (but rather pseudo-polynomial time), it still remains faster than browsing the Cartesian product of action-sets and is tractable on average,
as experimented in the appendix.
Symmetric games \cite{jiang2007computing} can also be generalized to MOGs:
\begin{definition}
In a \emph{multi-objective symmetric game}\index{multi-objective symmetric game}, 
individual payoffs are not impacted by
the agents' identities. There is one sole action-set $A^{\ast}$ for
every agent $i$. So, when deciding action $a^{\ast}\in A^{\ast}$, the
multi-objective reward only depends on the number of agents that decided every action.
Consequently, the game is not specified for every action-profile $\bm{a}\in A=\prod_{i\in N}A^{\ast}$
and every agent $i$, but rather for every action $a^{\ast}\in A^{\ast}$
and every \emph{configuration} $c:A^{\ast}\rightarrow\mathbb{N}$,
where number $c(a^{\ast})\in\mathbb{N}$ indicates the number of agents
deciding action $a^{\ast}$. Therefore, the utility
is given by a function $\bm{u}^{\ast}$ such that 
$\bm{u}^{\ast}(a^{\ast},c)\in\RE^{d}$
is the payoff for deciding action $a^{\ast}$ when
configuration $c$ occurs. 
\end{definition}
There is a number ${n+\alpha-1 \choose \alpha-1}$
of configurations\footnote{To enumerate the number of ways to distribute number $n$ of symmetric
agents into $\alpha$ parts, one enumerates the ways to choose $\alpha-1$
``separators'' in $n+\alpha-1$ elements.} to which the MO symmetric game associates MO vectors. As a consequence, generalizing to vectorial payoffs,
the representation length is $L=\alpha{n+\alpha-1 \choose \alpha-1}d$,
and when the numbers $\alpha$ and $d$ are fixed constant, length is
$L(n)\in\Theta\left(\alpha n^{\alpha}d\right)$. Quite simply, 
for computing $\CE$, $\EFF[\CE]$ and $\WST[\CE]$, configurations enumeration 
already takes polynomial time. 
\begin{theorem}
Given a multi-objective symmetric game with fixed $\alpha$,
\begin{itemize}
\item computing $\PE$ and $\CE$ takes time $O(n^{\alpha}\alpha^{2}d)=O(L\alpha)$; 
\item computing $\EFF[\CE]$ and $\WST[\CE]$ takes time $O(n^{2\alpha}d)=O(L^{2})$.
If $d=2$, this lowers to $O(L(\alpha+\log_{}(L)))$.
\end{itemize}
\end{theorem}

\subsection{Computing MO-CR}

In this subsection, we address the problem of computing the set $\text{MO-CR}[\CE,\CF]$,
given sets of worst equilibria outcomes $\WST[\CE]$ and efficient
outcomes $\CF$. Algorithm 1 (below) computes such set. In the algorithm, set $D^{t}$ denotes
a set of vectors. Given two vectors, $\bm{x},\bm{y}\in\REp$, 
let $\bm{x}\wedge \bm{y}$
denote the vector defined by $\forall k\in\OBJ,~(\bm{x}\wedge \bm{y})_{k}=\min\{x_{k},y_{k}\}$,
let $\bm{x}^{\bm{y}}\in\REp$ be the vector defined by $\forall k\in\OBJ, (\bm{x}^{\bm{y}})_k=(x_k)^{y_k}$,
and recall that $\forall k\in\OBJ,~(\bm{x}/\bm{y})_{k}=x_{k}/y_{k}$. 
\index{notations!$\bm{x}\wedge \bm{y}$, vectorial minimum of two vectors}
\index{notations!$\bm{x}^{\bm{y}}$, vectorial exponential of two vectors}

\begin{algorithm}[!h] 
\KwIn{$\WST[\CE]=\{\bm{y}^1,\ldots,\bm{y}^q\}$ and $\CF=\{\bm{z}^1,\ldots,\bm{z}^m\}$} \KwOut{$\text{MO-CR}=\EFF[R[\WST[\CE],\CF]]$} \ \\ [-1.5ex] 
{\bf create } $D^1\leftarrow \{\bm{y}^1/\bm{z}\in\REp~|~\bm{z}\in\CF\}$\\ 
\For{$t=2,\ldots,q$}{ 
	$D^t\leftarrow\EFF[\{\bm{\rho} ~\wedge~ (\bm{y}^t/\bm{z})~~|~~\bm{\rho}\in 		D^{t-1},~~\bm{z}\in\CF\}]$ 
} 
{\bf return } $D^q$  
\caption{Computing MO-CR in polynomial-time} 
\end{algorithm}

\begin{theorem}
Algorithm 1 outputs $\text{MO-CR}[\CE,\CF]$ in poly-time $O((qm)^{2d-1}d),$ where $q=|\WST[\CE]|$ and $m=|\CF|$ denote
the size of the inputs, and $d$ is fixed. 
\end{theorem}
\begin{proof}Algorithm
1 calculates product $\cap_{\bm{y}\in\WST[\CE]}\cup_{\bm{z}\in\CF}\CC(\bm{y}/\bm{z})$, where there could be $m^q$ terms in the output. This set-algebra of cone-unions is compact.
\end{proof}
A decisive corollary is that given
an MO game with length $L$ that satisfies $q=O(\text{poly}(L))$,
$m=O(\text{poly}(L))$ and both sets $\WST[\CE]$ and $\CF$ are computable
in time $O(\text{poly}(L))$, then one can compute $\text{MO-CR}$
in polynomial time $O(\text{poly}(L))$. For instance, it is the case
with MO normal forms or MO symmetric games. So this approach is not intractable in the most basic cases.

\subsection{Approximation of the MO-CR for MO compact representations}

Unfortunately, Algorithm 1 is not practical when the MO game has a compact form and cardinalities $q,m$ are exponentials with respect to the compact size of the game's representation. For instance, this is the case for multi-objective graphical games. Theorem \ref{th:approx} below answers this issue by taking only a small and approximate representation of sets $\WST[\CE]$ and $\CF$, in order to output a guaranteed approximation of sets $\text{MO-CR}$ or $R[\WST[\CE],\CF]$. This suggests the following general method:
\begin{enumerate}
\item Given a compact MOG representation, compute quickly an approximation  $E^{(\varepsilon)}$ of $\WST[\CE]$ and an approximation  $F^{(\varepsilon')}$ of $\CF$.
\item Then, given $E^{(\varepsilon)}$ and $F^{(\varepsilon')}$, use Algorithm 1 to approximate the MO-CR.
\end{enumerate}

For this general method to be implemented rigorously, we must specify the precise definitions of the two approximations required in input, for the desired output to be indeed some approximation of the MO-CR.


 Firstly, let us specify the {output}. 
The ratios in $R[\WST[\CE],\CF]$ must be represented, even approximately, but only by using valid ratios of efficiency, as below.\index{covering}
\begin{definition}[$(1+\varepsilon)$-covering]
\label{def:covering}
Given $R\subset\REp$ and $\varepsilon>0$,~~~ 
$R^{(\varepsilon)}\subset R$ is a $(1+\varepsilon)$-covering of $R$, if and only if:
$$
\forall \bm{\rho}\in R,\quad
\exists \bm{\rho}'\in R^{(\varepsilon)}:\quad
(1+\varepsilon)\bm{\rho}' \WPP \bm{\rho}
$$
\end{definition}
For instance, $R[\WST[\CE],\CF]$ is $(1+0)$-covered by $\text{MO-CR}=\EFF[R[\WST[\CE],\CF]]$.
Denote $\bm{\varphi}:\REp\rightarrow\NEp$ the discretization into the $(1+\varepsilon)$-logarithmic grid. Given a vector $\bm{x}\in\REp$, $\bm{\varphi}(x)$ is defined by: $\forall k\in\OBJ,~~\varphi_k(x)=\lfloor\log_{(1+\varepsilon)}(x_k)\rfloor$. 
A typical implementation of $(1+\varepsilon)$-coverings are the logarithmic $(1+\varepsilon)$-coverings, which consist in taking one vector of $R$ in each reciprocal image of $\bm{\varphi}(R)$. That is, for each $\bm{l}\in\bm{\varphi}(R)$, take one $\bm{\rho}$ in $\bm{\varphi}^{-1}(\bm{l})$.
The logarithmic grid is depicted in Fig. \ref{fig:mo:approx}.

Now we must specify rigorously what approximate representations $E^{(\varepsilon_1)}$ of set $\WST[\CE]$, and $F^{(\varepsilon_2)}$  of set $\CF$ we should take in input,  in order to guarantee that $R[E^{(\varepsilon_1)},F^{(\varepsilon_2)}]$ is an $(1+\varepsilon)$-covering of $R[\WST[\CE],\CF]$.
%
%
Definitions \ref{def:under:covering} and \ref{def:stick:covering} come from the need of specific approximate representations that will carry the guarantees to the approximate final output $R[E^{(\varepsilon_1)},F^{(\varepsilon_2)}]$.\index{under-covering}

\begin{definition}[$(1+\varepsilon)$-under-covering]
\label{def:under:covering}
Given $\varepsilon>0$, $E\subset\REp$ and $E^{(\varepsilon)}\subset\REp$,
$E^{(\varepsilon)}$ $(1+\varepsilon)$-{under}-covers $E$ if and only if:
\begin{eqnarray*}
\forall \bm{y}\in E,~~
\exists \bm{y}'\in E^{(\varepsilon)} &:&
\bm{y}\WPP \bm{y}'\\
\text{and}~~~
\forall \bm{y}'\in E^{(\varepsilon)},~~
\exists \bm{y}\in E &:&
(1+\varepsilon)\bm{y}'\WPP \bm{y}
\end{eqnarray*}
\end{definition}

The first condition states that $E^{(\varepsilon)}$ bounds $E$ from below.
The second condition states that this lower bound is precise within a multiplicative $(1+\varepsilon)$. Given $E$, one can implement Definition \ref{def:under:covering} by using the log-grid (see e.g. Fig. \ref{fig:mo:approx}):
$$
E^{(\varepsilon)}\leftarrow\WST\left[~~\left\{~\bm{e}^{\bm{l}}\in\REp\mid \bm{l}\in\bm{\varphi}\left(\WST[\CE]\right)\right\}~~\right]
$$
where $\bm{\varphi}(\WST[\CE])=\{\bm{\varphi}(\bm{y})\in\NEp\mid\bm{y}\in \WST[\CE]\}$, and given $\bm{l}\in\NEp$, the vector $\bm{e}^{\bm{l}}$ is defined by $(\bm{e}^{\bm{l}})_k=(1+\varepsilon)^{l_k}$.
Now let us state what approximation is required on the set of efficient outcomes $\CF$.\index{stick-covering}
\begin{definition}[$(1+\varepsilon)$-stick-covering]
\label{def:stick:covering}
Given $\varepsilon>0$, $F\subset\REp$ and $F^{(\varepsilon)}\subset\REp$, 
$F^{(\varepsilon)}$ $(1+\varepsilon)$-{stick}-covers $F$ if and only if:
\begin{eqnarray*}
\forall \bm{z}'\in F^{(\varepsilon)},~~
\exists \bm{z}\in F &:&
\bm{z}'\WPP \bm{z}\\
\text{and}~~~
\forall \bm{z}\in F,~~
\exists \bm{z}'\in F^{(\varepsilon)} &:&
(1+\varepsilon)\bm{z}\WPP \bm{z}'
\end{eqnarray*}
\end{definition}
The first condition is easily satisfiable by $F^{(\varepsilon)}\subseteq F$.
The second condition states that $F^{(\varepsilon)}$ sticks to $F$. Given $F$, one can implement Definition \ref{def:stick:covering} as in Figure \ref{fig:mo:approx}: Take one element of $\CF$ per cell of the logarithmic grid, and then take $\WST$ of this set of elements.
Now we can state that with an approximate Phase 1, the precision transfers to Phase 2 in polynomial time, as follows.

\begin{figure}[t]
\caption{MO approximations, depictions of under and stick coverings}
\label{fig:mo:approx}
\begin{minipage}{0.5\textwidth}
\centering
\includegraphics[scale=0.5]{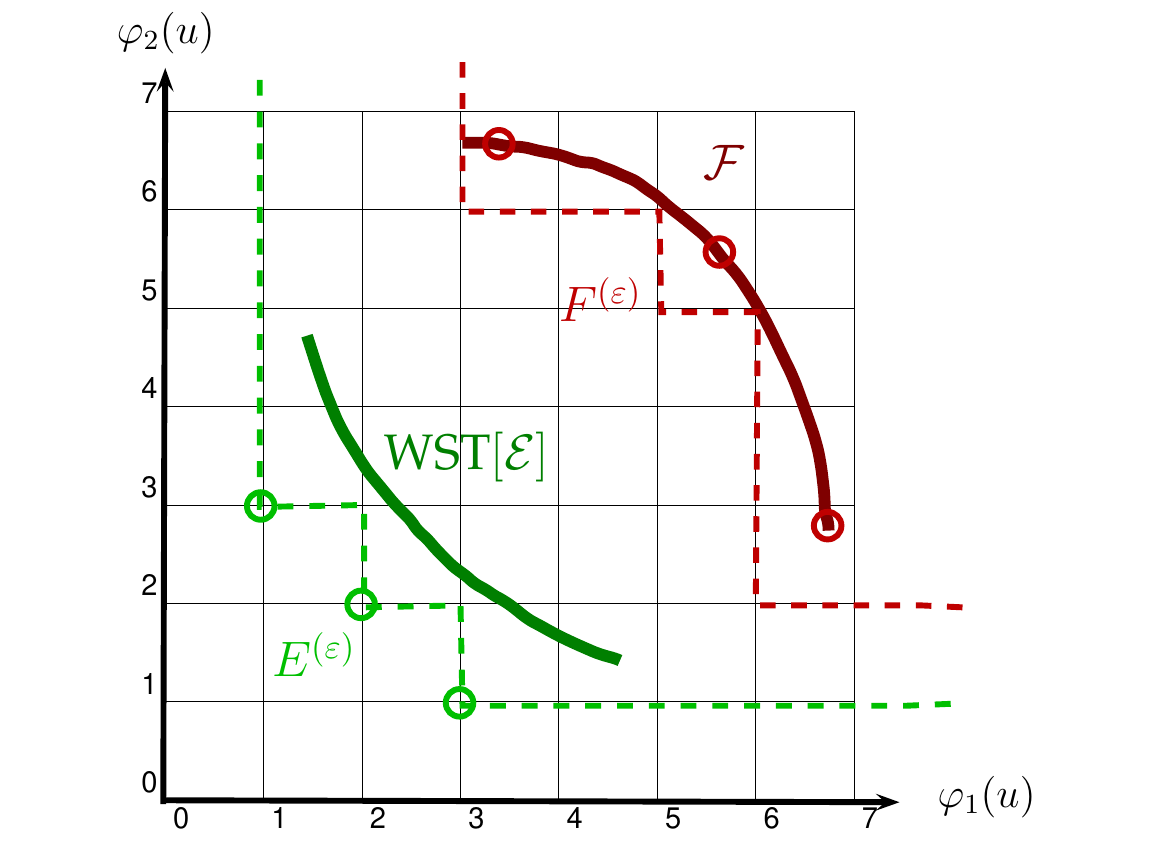}
\end{minipage}
~
\begin{minipage}{0.48\textwidth}
\small
$E^{(\varepsilon)}$ (the green dots below $\WST[\CE]$) is a $(1+\varepsilon)$-under-covering of set $\WST[\CE]$.\\[3ex]
$F^{(\varepsilon)}$ (the three red dots in $\CF$) is a $(1+\varepsilon)$-stick-covering of the dark-red set $\CF$.
\end{minipage}

\end{figure}

\begin{lemma}\label{lem:approx} Given  $\varepsilon_1,\varepsilon_2>0$ and approximations $E$ of $\CE$ and $F$ of $\CF$, if
\begin{eqnarray}
\forall \bm{y}\in\CE, \exists \bm{y}'\in E,\quad \bm{y}\WPP \bm{y}' &
\quad \text{and}\quad & 
\forall \bm{y}'\in E, \exists \bm{y}\in\CE,\quad (1+\varepsilon_1)\bm{y}'\WPP \bm{y}\label{eq:approx:E}\\
\forall \bm{z}'\in F, \exists \bm{z}\in\CF,\quad \bm{z}'\WPP \bm{z} &
\quad \text{and}\quad & 
\forall \bm{z}\in\CF, \exists \bm{z}'\in F,\quad (1+\varepsilon_2)\bm{z}\WPP \bm{z}'\label{eq:approx:F}
\end{eqnarray}
holds, then it follows that $R[E,F]\subseteq R[\CE,\CF]$ and:
\begin{eqnarray}
\forall \bm{\rho}\in R[\CE,\CF],\quad \exists \bm{\rho}'\in R[E,F],\quad (1+\varepsilon_1)(1+\varepsilon_2)\bm{\rho}'\WPP\bm{\rho}\label{eq:approx:R}
\end{eqnarray}
\end{lemma}

Equations (\ref{eq:approx:E}) and  (\ref{eq:approx:F}) state approximation bounds as in Definitions \ref{def:under:covering} and \ref{def:stick:covering}. Equations (\ref{eq:approx:E}) state that $(1+\varepsilon_1)^{-1}\CE$ bounds below $E$ which bounds below $\CE$. Equations (\ref{eq:approx:F}) state that $\CF$ bounds below $F$ which bounds below $(1+\varepsilon_2)\CF$. Crucially, whatever the sizes of $\CE$ and $\CF$, there exist such approximations $E$ and $F$ with respective sizes $O((1/\varepsilon_1)^{d-1})$ and $O((1/\varepsilon_2)^{d-1})$  \cite{papadimitriou2000approximability}, yielding the approximation scheme below.

\begin{theorem}[Approximation Scheme for MO-CR]\label{th:approx}
Given a compact MOG of representation length $L$, precisions $\varepsilon_1,\varepsilon_2>0$  and two algorithms  to compute approximations $E$ of $\CE$ and $F$ of $\CF$ in the sense of Equations (\ref{eq:approx:E}) and (\ref{eq:approx:F}) that take time $\theta_{\CE}(\varepsilon_1,L)$ and $\theta_{\CF}(\varepsilon_2,L)$, one can approximate $R[\CE,\CF]$ in the sense of Equation (\ref{eq:approx:R}) in time
$O\left(\theta_{\CE}(\varepsilon_1,L) \quad+\quad \theta_{\CF}(\varepsilon_2,L) \quad+\quad {(\varepsilon_1 \varepsilon_2)^{-(d-1)(2d-1)}}\right)$.
\end{theorem}
For MO graphical games, Phase 1 could be instantiated with approximate junction-tree algorithms on MO graphical models \cite{dubus2009multiobjective}. For MO symmetric action-graph games, in the same fashion, one could generalize existing algorithms \cite{jiang2007computing}. More generally, for the worst equilibria $\WST[\CE]$ and the efficient outcomes $\CF$, one could also use meta-heuristics with experimental guarantees.

\section{Conclusion: discussion and prospects}

Along with equilibrium existence, potential functions also usually
guarantee the convergence of best-response dynamics. This easily generalizes
to dynamics where every deviation step is an individual Pareto-improvement.
However, when studying a dynamics based on a refinement of the Pareto-dominance,
convergence is not always guaranteed. 

Pareto-Nash equilibria, which encompass the possible outcomes of MO
games, very likely exist.  The precision
of PN-equilibria inevitably relies on the uncertainty on preferences.
A promising research path would be to linearly constrain the utility
functions of agents. This would induce a polytope and would boil down 
to another MO game where
every objective corresponds to an extreme point of the induced polytope.
%
%
The efficiency of several multi-objective games could be analyzed by using the contributions in this paper.


\bibliographystyle{splncs04}
\bibliography{newMOG}

\section{Proof of Theorem 1}

Let $0_d\in\RE^d$ denote the $d$-dimensional MO vector with $d$ zero components. Let $a\in A$ be an action-profile. To state that $a$ is a $\PE$-equilibrium is equivalent to state that for every agent $i$ and every individual deviation $b^i\in A^i$, it holds that:
$$
u^i(b^i,a^{-i})
\quad\not\PP\quad
u^i(a)
$$
From the definition of a potential $\Phi$, it is equivalent to state that, for every agent $i$ and every individual deviation $b^i\in A^i$, it holds that:
$$
\Phi(b^i,a^{-i})-\Phi(a)
\quad =\quad
u^i(b^i,a^{-i})-u^i(a)
\quad\not\PP\quad
0_d
$$
That is, $\Phi(b^i,a^{-i})\not\PP\Phi(a)$, which means that $a\in\LOC(\Phi)$.

Furthermore, the existence of local optima for the potential function generalizes to the MO case: The set of locally-Pareto-efficient action-profiles is necessarily non-empty, otherwise, given $t\in\mathbb{N}$, whatever the action-profile $a_{(t)}$, one could always find an action-profile $a_{(t+1)}$ in its neighbourhood of individual deviations, such that $\Phi(a_{(t+1)})\PP\Phi(a_{(t)})$. Therefore, one could build an infinite sequence $(a_{(t)})_{t\in\mathbb{N}}$ such that $\Phi(a_{(t+1)})\PP\Phi(a_{(t)})$; and since the Pareto-dominance $\PP$ is a strict partial order and $\Phi$ a (deterministic) function, one would have an infinite number of distinct action-profiles, contradicting the fact that $|A|\leq \alpha^n$ is finite.

\section{Proof of Theorem 2}

We will denote by $\PROBA_{n,\alpha,\pi}$ the probability distribution that draws a normal form game (SO or MO) with $n$ agents, $\alpha=\alpha^i=|A^i|$ actions-per-agent, and the payoffs $u^i(a)$ according to the distribution $\pi$ on $\RE$ or $\RE^d$. Also, according to $\PROBA_{n,\alpha,\pi}$, given an agent $i$ and an action-profile $a=(a^i,a^{-i})$, let us denote by $X_{i,a}\in\{0,1\}$ the random variable (RV) which is equal to $1$ if and only if for agent $i$, the action $a^i$ is a best response (or efficient response) to the adversary action-profile $a^{-i}$. Given an action-profile $a$, let us denote by $Y_{a}=\min_{i\in N}\{X_{i,a}\}$ the binary RV which is equal to $1$ if and only if the action-profile $a$ is a PN equilibrium. Finally, let $Z=\sum_{a\in A} Y_{a}$ denote the number of pure Nash-equilibria action-profiles in the game. For simplicity, we may use the name of a binary random variable as a shorthand for the event that this RV equals $1$. Since for every agent $i$ and every adversary action-profile $a^{-i}$ there is (almost surely) only one best-response $b^i\in A^i$ in $u^i(A^i,a^{-i})$ (because payoffs are almost surely different), an IID uniform distribution on $[0,1]$ amounts to whatever IID distribution that will almost surely draw uniformly one single best-response in $u^i(A^i,a^{-i})$.

\textit{Generalization to multi-objective.} While in the SO case, there is almost surely only one best-response, when considering MO games, the main technical difference lies in the average number of ``best-responses'' (or here, Pareto-efficient responses) which is in most cases greater than $1$, due to the $(d-1)$ dimensional surface-like shape of the Pareto-efficient set in $\RE^d$. For instance, it can be shown that when drawing a number $\alpha$ of MO payoffs according to a uniform distribution on the simplex $\mathcal{S}_{\OBJ}=\{u\in\RE_+^d\mid \sum_{k\in\OBJ}u_k\leq 1\}$, then by counting the vectors on the outer face, the number $\beta$ of Pareto-efficient vectors among the $\alpha$ vectors satisfies: 
\begin{eqnarray}
\ESP[\beta] & \quad \sim \quad &
\frac{d}{(d!)^{1/d}}~~\alpha^{\frac{d-1}{d}}
\quad\quad\quad\text{as}\quad
\alpha\rightarrow\infty
\label{eq:integfac}
\end{eqnarray}
in the sense that the ratio of the left and right members of $\sim$ tends to $1$. As a consequence, on the simplex $\mathcal{S}_{\OBJ}$, one quickly has a number of Pareto-efficient responses $\beta$ strictly greater than $1$ as the number of actions $\alpha$ grows. (The number of objectives $d$ is fixed.)
 
For the sake of simplicity, we then assume a probability distribution $\PROBA_{n,\alpha,\beta}$, that builds randomly an $n$-agents normal form game with $\alpha$ actions-per-agent. For the sake of simplicity, for every agent $i$, and every adversary action-profile $a^{-i}\in\prod_{j\neq i} A^j$, there is a fixed number $\beta:1\leq\beta\leq\alpha$ of Pareto-efficient responses   (supposedly, according to some vectorial payoffs $u^i(A^i,a^{-i})$ selected independently and uniformly at random in $A^i$). Recall that the number $\beta$ can be reasonably supposed greater than $1$ (see Equation \ref{eq:integfac}).

\medskip{}

\cite{bienayme1867} Recall that the Bienaym\'e-Tchebychev inequality states that for a random variable $Z$ with expectancy $\ESP[Z]$ and variance $\VAR[Z]$, for every parameter $\mu\in\RE_+$ it holds that:
\begin{eqnarray*}
\PROBA(|Z-\ESP[Z]|\geq \mu) &\leq & \frac{\VAR[Z]}{\mu}
\end{eqnarray*}
In simple words, a random variable is unlikely to spread more than its variance. 

\medskip{}

Let us now study the expectation of the number of PN-equilibria $\ESP_{n,\alpha,\beta}[Z]$. One has:  \begin{eqnarray} \ESP_{n,\alpha,\beta}\left[Z\right]  &=&  \ESP_{n,\alpha,\beta}\left[\sum_{a\in A} Y_a \right] \label{eq14}\\ &=&  \sum_{a\in A} \ESP_{n,\alpha,\beta}\left[ Y_a \right] \label{eq15}\\ &=&  \sum_{a\in A} \ESP_{n,\alpha,\beta}\left[ \min_{i\in N} X_{i,a} \right] \label{eq16}\\ &=&  \sum_{a\in A} \PROBA_{n,\alpha,\beta}\left(\wedge_{i\in N} \{X_{i,a}\} \right) \label{eq17}\\ &=&  \sum_{a\in A} \prod_{i\in N} \PROBA_{n,\alpha,\beta}\left( X_{i,a} \right) \label{eq18}\\ &=&  \sum_{a\in A} \prod_{i\in N} \frac{\beta}{\alpha} \label{eq19} \\ &=&  \alpha^n (\beta/\alpha)^n \label{eq110}\\ &=& \beta^n \label{eq111} \end{eqnarray} Equation (\ref{eq14}) uses the definition of the RV $Z$. Equation (\ref{eq15}) uses the linearity of expectation. Equation (\ref{eq16}) uses the definition of the RV $Y_a$. Equation (\ref{eq17}) formulates it as an event. Equation (\ref{eq18}) uses the independence of payoffs between agents. Equation (\ref{eq19}) uses the definition the probability $\PROBA_{n,\alpha,\beta}$: uniform. Equation (\ref{eq110}) uses that $|A|=\alpha^n$ and that $\prod_{i\in N}(\beta/\alpha)=(\beta/\alpha)^n$. Equation (\ref{eq111}) concludes that: $\ESP[Z]=\beta^n$. Therefore, the number of PN-equilibria $Z$ is in expectation an exponential of basis $\beta$ with respect to the number of agents $n$. Let us now study the variance of the number of PN-equilibria $Z$: \begin{eqnarray}  & &\VAR\left(Z\right) \\ &=& \VAR\left(\sum_{a\in A} Y_a \right)\label{eq:112}\\ &=& \sum_{a\in A} \sum_{b\in A} \COV(Y_a,Y_b)\label{eq:113}\\ &=& \sum_{a\in A} \sum_{b\in A} \ESP[Y_a Y_b] - \ESP[Y_a] \ESP[Y_b]\label{eq:114}\\ &=& \sum_{b\in A} \sum_{a\in A} \left(\prod_{i\in N}\PROBA(X_{i,a} X_{i,b}) - \left(\frac{\beta}{\alpha}\right)^{2n}\right)\label{eq:115} \end{eqnarray} Equation (\ref{eq:112}) uses the definition of the RV $Z$. Equation (\ref{eq:113}) is the variance of the sum of RVs $\sum_{a\in A} Y_a$. Equation (\ref{eq:114}) uses the definition of the covariance $\COV(Y_a,Y_b)$. In Equation (\ref{eq:115}) the first terms $\ESP[Y_a Y_b]=\prod_{i\in N}\PROBA(X_{i,a} X_{i,b})$ result from the independences of payoffs between players. The second terms $\ESP[Y_a]\ESP[Y_b]=(\beta/\alpha)^{2n}$ result from the same calculus as for the expectation $\ESP[Z]$. Remark that by symmetry, all the $\alpha^n$ terms of the outer sum are equal. Fixing an action-profile $b\in A$, let us continue this calculus below: \begin{eqnarray} \VAR\left(Z\right) &=&  \alpha^n \sum_{a\in A} \left(\prod_{i\in N}\PROBA(X_{i,a} X_{i,b}) - \left(\frac{\beta}{\alpha}\right)^{2n}\right)\label{eq:116} \end{eqnarray}Now, having fixed an action-profile $b\in A$, given an action-profile $a$ and an agent $i$, let us study the value of the probability $\PROBA(X_{i,a}, X_{i,b})$. Remark that it will depend on whether the random variables $X_{i,a}$ and  $X_{i,b}$ are independent or not: \begin{itemize} \item If $a^{-i}\neq b^{-i}$, then the payoffs are independent, and one has the probability: $$\PROBA(X_{i,a}, X_{i,b})=\PROBA(X_{i,a})\PROBA(X_{i,b})=(\beta/\alpha)^2$$ \item If $a^{-i}= b^{-i}$ with $a^i\neq b^i$, then the payoffs are dependent, and one has the probability: $$\PROBA(X_{i,a}, X_{i,b})=\PROBA(X_{i,a}~|~X_{i,b}) \PROBA(X_{i,b})=\frac{(\beta-1)\beta}{\alpha^2}$$ \item Finally, if $a=b$, then $\PROBA(X_{i,a}, X_{i,b})=\PROBA(X_{i,a})=\beta/\alpha$. \end{itemize}Now, (having fixed an action-profile $b\in A$) let us study the terms in the sum $\sum_{a\in A}$. Given an action-profile $a\in A$, one has:\\[2ex] $\bullet$ If $a=b$, which occurs exactly once, then $\PROBA(X_{i,a}, X_{i,b})=\beta/\alpha$, and the term $\prod\nolimits_{i\in N}\PROBA(X_{i,a} X_{i,b}) - \left(\beta/\alpha\right)^{2n}$ equals $(\beta/\alpha)^n - \left(\beta/\alpha\right)^{2n}$.\\[2ex] $\bullet$ If for some agent $i$, it holds that $a^{-i}= b^{-i}$ with $a^i\neq b^i$, then a distinct agent $j$ cannot satisfy $a^{-j}= b^{-j}$, because of $a^i\neq b^i$; hence the other agents (other than agent $i$) fall into the case of $a^{-j}\neq b^{-j}$. This occurs exactly $n(\alpha-1)$ times, and then while it holds that $\PROBA(X_{i,a}, X_{i,b})=(\beta-1)\beta/\alpha^2$ for agent $i$, for the other agents $j$, it holds that $\PROBA(X_{j,a}, X_{j,b})=(\beta/\alpha)^2$. Therefore, the term $\prod\nolimits_{i\in N}\PROBA(X_{i,a} X_{i,b}) - \left(\beta/\alpha\right)^{2n}$ equals $((\beta-1)\beta/\alpha^2)(\beta/\alpha)^{2n-2} - \left(\beta/\alpha\right)^{2n}$, that is:\\[2ex] \begin{eqnarray*} \prod\nolimits_{i\in N}\PROBA(X_{i,a} X_{i,b}) - \left(\beta/\alpha\right)^{2n} &=&  \frac {(\beta-1)\beta^{2n-1}-\beta^{2n}} {\alpha^{2n}}\\ &=& \frac {-\beta^{2n-1}} {\alpha^{2n}} \end{eqnarray*} $\bullet$ In the last case, if for every agent $i$, it holds that $a^{-i}\neq b^{-i}$, then the term cancels.\\[2ex] To conclude, the variance of the number of Pareto-Nash equilibria is: \begin{eqnarray} &&\VAR\left(Z\right) \\ &=&  \alpha^n \sum_{a\in A} \left(\prod\nolimits_{i\in N}\PROBA(X_{i,a} X_{i,b}) \quad-\quad \left(\beta/\alpha\right)^{2n}\right)\label{eq:117}\\ &=&  \alpha^n \left((\beta/\alpha)^n - (\beta/\alpha)^{2n}-n(\alpha-1)\beta^{2n-1}/ \alpha^{2n}\right)\label{eq:118}\\ &=& \beta^n \left( 1 - (\beta/\alpha)^n - n(\alpha-1)\beta^{n-1}/\alpha^{n}\right)\\ &=& \beta^n \left( 1 - \quad(\beta/\alpha)^n(1+n(\alpha-1)/\beta)\quad\right)\\ &\leq & \beta^n  \end{eqnarray} To finish, since we have an expectation $\ESP[Z]=\beta^n$ and a variance $\VAR(Z)\leq\beta^n$, a straightforward use of the Bienaym\'e-Tchebychev inequality concludes that for any given number $\gamma\in(0,1)$, it holds that: \begin{eqnarray} \PROBA(|Z-\beta^n|\leq \gamma \beta^n) &\geq & 1 - \frac{\beta^n}{\gamma^2 \beta^{2n}}\\ & = & 1 - \frac{1}{\gamma^2 \beta^{n}} \end{eqnarray}

\section{Proofs concerning the properties of the multi-objective coordination ratio}
%
%

We show that Definition \ref{def:MOCR} satisfies the following.
Given $\CE,\CF\subset\REp$ and $\bm{r}\in\REp$: 
\begin{eqnarray} 
\mbox{MO-CR}[\CE,\CF] 
&\quad\subseteq& 
\REp\label{eq:ratio:6:bis}\\ 
\mbox{MO-CR}[\{\bm{0}\},\CF] 
&\quad=& 
\{\bm{0}\}\label{eq:ratio:7:bis}\\ 
\mbox{MO-CR}[\bm{r}\star\CE,\CF] 
& \quad=& 
\bm{r}\star\mbox{MO-CR}[\CE,\CF]\label{eq:ratio:8:bis}\\ 
\mbox{MO-CR}[\CE,\bm{r}\star\CF] 
& \quad=& 
\mbox{MO-CR}[\CE,\CF]/\bm{r}\label{eq:ratio:9:bis}\\ 
\CE\subseteq\CF 
& \quad\Leftrightarrow& 
\bm{1}\in \mbox{MO-CR}[\CE,\CF]\label{eq:ratio:10:bis} 
\end{eqnarray}

Property (\ref{eq:ratio:6:bis})\quad By definition, set $\mbox{MO-CR}[\CE,\CF]$ is a set of vectors in $\REd$.

Property (\ref{eq:ratio:7:bis})\quad If $\CE=\{\bm{0}\}$, then the condition $\rho\in R[\CE,\CF]$, which is $\forall \bm{y}\in\CE,$ $\quad\exists \bm{z}\in\CF,$ $\bm{y}/\bm{z}\WPP\bm{\rho}$, rewrites $\bm{0}\WPP \bm{\rho}$. Then $\EFF[\CC(\bm{0})]=\{\bm{0}\}$.

Property (\ref{eq:ratio:8:bis})\quad We just need to show that $R[\bm{r}\star\CE,\CF]=\bm{r}\star R[\CE,\CF]$. Condition $\rho\in R[\bm{r}\star\CE,\CF]$ rewrites into $\forall \bm{y}\in\CE,\quad\exists \bm{z}\in\CF,\quad \bm{r}\star\bm{y}/\bm{z}\WPP\bm{\rho}$. Then one has $\rho\in \bm{r}\star R[\CE,\CF]$. The converse also holds by a similar argument.

Property (\ref{eq:ratio:9:bis})\quad Similarly, one can show that $R[\CE,\bm{r}\star\CF]=R[\CE,\CF]/\bm{r}$.

Property (\ref{eq:ratio:10:bis})\quad First, note that since $\CF$ dominates $\CE$, it is not possible to have $\bm{\rho}\PP\bm{1}$ in $R[\CE,\CF]$. Second, for every $\bm{y}\in\CE$, one can then take $z=y$, and since $\bm{1}/\bm{1}\WPP \bm{1}$, one has $\bm{1}\in R[\CE,\CF]$. One can also show that if $\CE\not\subseteq\CF$ then $\bm{1}\not\in \mbox{MO-CR}[\CE,\CF]$.

\section{Proof of Theorem 3}

The computation of the best equilibria outcomes $\EFF[\CE]$ can be achieved by (1) computing the PN equilibria $\PE\subseteq A$, then (2) computing the equilibria outcomes $\CE=u(\PE)\subseteq \RE^d$ and finally (3) computing the best equilibria outcomes $\EFF[\CE]\subseteq\CE$ (or the worst ones  $\WST[\CE]\subseteq\CE$). 

(1) For this purpose, for every agent $i\in N$ and each adversary action profile $a^{-i}\in A^{-i}$, one has to compute which individual actions give a Pareto-efficient evaluation in $u^i(A^i,a^{-i})$ (which takes time $O(\alpha^2d)$, or if $d=2$ then $O(\alpha\log_2(\alpha))$), in order to mark which action-profiles can be a PN equilibrium from $i$'s point of view. Hence, computing $\PE$ takes time $O(n \alpha^{n-1} \alpha^2 d)$ (or if $d=2$ $O(n \alpha^{n} \log_2(\alpha) )$). In the worst case, $\PE=A$ hence $|\PE|=O(\alpha^n)$. 

Then, (2) computing the image through total-utilitarianism $\CE=u(\PE)$ requires for each $a\in \PE$ the addition of $n$ vectors, in time $nd|\PE|=O(n\alpha^n d)$. 

(3) Finally, the computation of $\EFF[\CE]$ given $\CE$ takes time $O(|\CE|^2 d)=O(\alpha^{2n} d)$; and the same holds for $\WST[\CE]$. To sum up, the computation of $\EFF[\CE]$ (or of $\WST[\CE]$) takes time $O(n \alpha^{n+1} d + \alpha^{2n} d)$. If $d=2$, this significantly lowers to $O(n\alpha^n \log_2(\alpha))$, by using a data structure (e.g. an AVL tree) that orders vectors according to the first objective and does comparisons on the second objective.

\section{Proof of Theorem 4}

Since the game is symmetric, every configuration $c$ represents an equivalence class in the set of action-profiles $A$; hence, a set of configurations represents a subset of the action-profiles. Therefore, in order to compute the set of Pareto-Nash equilibria $\PE\subseteq A$, a set of configurations is an acceptable output and even a more compact one. The problem to decide if a given configuration $c$ is a (pure-strategy) Pareto-Nash equilibrium is easy: one only has to test for every action $a^\ast\in A^\ast$ such that\footnote{That is such that the action $a^\ast$ is decided by someone.} $c(a^\ast)\geq 1$, if that action is a Pareto-efficient individual decision. An individual deviation to another action $b^\ast\in A^\ast$ induces the configuration $c'$ obtained from the configuration $c$ by subtracting 1 from the number $c(a^\ast)$ and adding 1 to the number $c(b^\ast)$ of agents deciding the action $b^\ast$. Therefore, testing if a configuration $c$ is a $\PE$ equilibrium takes time $O(\alpha^2 d)$. As a consequence, the computation of the set of $\PE$ equilibria takes times $O(n^\alpha \alpha^2 d)$, that is poly-time $O(L\alpha)$. Also, computing an utilitarian evaluation $u(c)=\sum_{a^\ast\in A^\ast}c(a^\ast)u^{\ast}(a^\ast,c)\in\RE^d$ requires $O(\alpha)$ multiplications and additions; hence computing the set of equilibria outcomes $\CE=u(\PE)$ (starting from the set $\PE$ which size is $O(L)$) also takes poly-time $O(L\alpha)$.
Since the number of equilibria outcomes is bounded by the number of configurations,  it follows that computing the sets of best and worst equilibria $\EFF[\CE]$ and $\WST[\CE]$ takes time $O(n^{2\alpha} d)$, that is poly-time $O(L^2)$.

\section{Proof of Theorem 5}

In order to compute $\text{MO-CR}=\EFF[R[\WST[\CE],\CF]]$, let us study the structure of $\bigcap_{y\in\WST[\CE]}\bigcup_{z\in\CF}\CC(y/z)$, by restricting a set-algebra to the following objects: 

\begin{definition}[Cone-Union] \label{def:coneunion} For set of vectors $X\subseteq\REp$, Cone-Union $\CC(X)$ is: $$ \CC(X) \quad=\quad\bigcup_{x\in X}\CC(x) \quad=\quad\{y\in\REp ~~|~~ \exists x\in X, x\WPP y\} $$  Let $\CC$ denote the set of all cone-unions of $\REp$. \end{definition} 

To define an algebra on $\CC$, one can supply $\CC$ with $\cup$ and $\cap$. \begin{lemma}[On the Set-Algebra $(\CC,\cup,\cap)$]\label{prop:algebra}~\\ Given two descriptions of cone-unions $X^1,X^2\subseteq\REp$, we have: $$\CC(X^1)\cup\CC(X^2)\quad=\quad\CC( X^1\cup X^2 )$$ Given two descriptions of cones $x^1,x^2\in\REp$, we have: $$\CC(x^1)\cap\CC(x^2)\quad=\quad\CC(x^1\wedge x^2)$$ where $x^1\wedge x^2\in\REp$ is: $\forall k\in\OBJ, (x^1\wedge x^2)_k=\min\{x^1_k,x^2_k\}$.\\ Given two descriptions of cone-unions $X^1,X^2\subseteq\REp$, we have: \begin{eqnarray*} \CC(X^1)\cap\CC(X^2) &=&\left(\cup_{x^1\in X^1}\CC(x^1)\right)\cap\left(\cup_{x^2\in X^2}\CC(x^2)\right)\\ &=&\bigcup_{(x^1,x^2)\in X^1\times X^2}\CC(x^1)\cap\CC(x^2)\\ &=&\bigcup_{(x^1,x^2)\in X^1\times X^2}\CC(x^1\wedge x^2)\\ &=&\CC( X^1\wedge X^2 ) \end{eqnarray*} where $X^1\wedge X^2=\{x^1\wedge x^2~|~x^1\in X^1,~x^2\in X^2\}\subseteq\REp$.\\ Therefore, $(\CC,\cup,\cap)$ is stable, and then is a set-algebra. \end{lemma} \begin{proof} The three properties derive from set calculus. \end{proof} The main consequence of Lemma \ref{prop:algebra} is that  $R[\WST[\CE],\CF]=\cap_{y\in\WST[\CE]}\cup_{z\in\CF}\CC(y/z)$ is a cone-union. Moreover, one can do the expansion for $\cap_{y\in\WST[\CE]}\cup_{z\in\CF}\CC(y/z)$ within the cone-unions, using expansions.\begin{remark}\label{rk:app:cone} For a finite set $X\subseteq\REp$, we have: $\CC(X)=\CC(\EFF[X])$. \end{remark}   \begin{proof} Firstly, we prove $\CC(X)\subseteq\CC(\EFF[X])$. If $y\in\CC(X)$, then there exists $x\in X$ such that $x\WPP y$. There are two cases, $x\in\EFF[X]$ and $x\not\in\EFF[X]$. If $x\in\EFF[X]$, then $y\in\CC(\EFF[X])$, by definition of a cone-union. Otherwise, if $x\not\in\EFF[X]$, then there exists $z\in X$ such that $z\PP x$. And since $X$ is finite, we can find such a $z$ in $\EFF[X]$, by iteratively taking $z'\PP z$ and $z\leftarrow z'$, until $z'\in\EFF[X]$, which will happen because $X$ is finite and $\PP$ is transitive and irreflexive. Hence, there exists $z\in\EFF[X]$ such that $z\PP x\WPP y$ and then $z\PP y$. Consequently, $y\in\CC(\EFF[X])$, by definition of a cone-union.
Conversely, $Y\subseteq X\Rightarrow \CC(Y)\subseteq\CC(X)$  proves $\CC(\EFF[X])\subseteq\CC(X)$. \end{proof}As a consequence of Remark \ref{rk:app:cone}, for $x\in\REp$, a simple cone $\CC(x)$ is fully described by its apex $x$. The main consequence of this remark is that $\CC(X)$ can be fully described and represented by $\EFF[X]$. For instance, since $R[\WST[\CE],\CF]$ is a cone-union (thanks to Lemma \ref{prop:algebra}), and since $\text{MO-CR}=\EFF[R[\WST[\CE],\CF]]$ (by definition of the MO-CR), then $R[\WST[\CE],\CF]$ is fully represented (as a cone-union) by the MO-CR, which means that $R[\WST[\CE],\CF]=\CC(\text{MO-CR})$.\\
 Recall that $q=|\WST[\CE]|$ and $m=|\CF|$. In this subsection, we also denote $\WST[\CE]=\{y^1,\ldots,y^q\}$ and $\CF=\{z^1,\ldots,z^m\}$. Let $\CA_q^m$ denote the set of functions $\pi$ from $\{1,\ldots,q\}$ to $\{1,\ldots,m\}$. (We have: $|\CA_q^m|=m^q$.) \begin{corollary}[The cone-union of MO-CR]~\\ Given $\WST[\CE]=\{y^1,\ldots,y^q\}$ and $\CF=\{z^1,\ldots,z^m\}$, we have: $$R[\WST[\CE],\CF]=\bigcup_{\pi\in\CA_q^m}\bigcap_{t=1}^{q} \CC(y^t / z^{\pi(t)})$$ and therefore: $$\text{MO-CR}[\CE,\CF]\quad=\quad\EFF\left[\left\{\bigwedge\nolimits_{t=1}^{q}y^t / z^{\pi(t)}~~|~~\pi\in\CA_q^m\right\}\right]$$ \end{corollary} \begin{proof} For the first statement, just think of an expansion. We write down $$R[\WST[\CE],\CF]=\cap_{y\in\WST[\CE]}\cup_{z\in\CF}\CC(y/z)$$ into the layers just below. There is one layer per $y^t$ in $\WST[\CE]=\{y^1,\ldots,y^t,\ldots,y^q\}$: $$ \begin{array}{ccccccccccl} & ( & \CC(\frac{y^1}{z^1}) & \cup & \CC(\frac{y^1}{z^2}) & \cup & \ldots & \cup & \CC(\frac{y^1}{z^m})& )&\text{layer 1}\\ \bigcap & ( & \CC(\frac{y^2}{z^1}) & \cup & \CC(\frac{y^2}{z^2}) & \cup & \ldots & \cup & \CC(\frac{y^2}{z^m})& )&\text{layer 2}\\ &&&&&\vdots\\ \bigcap & ( & \CC(\frac{y^q}{z^1}) & \cup & \CC(\frac{y^q}{z^2}) & \cup & \ldots & \cup & \CC(\frac{y^q}{z^m})& )&\text{layer q} \end{array} $$ Imagine the simple cones $\CC(\frac{y^t}{z^{\pi(t)}})$ as vertices and imagine edges going from each vertex of layer $t$ to each vertex of the next layer $(t+1)$. Let the function $\pi:\{1,\ldots,q\}\rightarrow\{1,\ldots,m\}$ denote a path from layer $1$ to layer $q$, where $\pi(t)$ is the vertex chosen in layer $t$. The expansion into a union outputs as many intersection-terms as paths from the first layer to the last one.   Consequently, in the result of the expansion into an union, each term is an intersection $\bigcap_{t=1}^{q} \CC(y^t / z^{\pi(t)})$. Then one has:
\begin{eqnarray*} R[\WST[\CE],\CF] &=&\bigcup_{\pi\in\CA_q^m}\bigcap_{t=1}^{q} \CC(y^t / z^{\pi(t)})\\ &=&\bigcup_{\pi\in\CA_q^m} \CC\left(\bigwedge_{t=1}^{q} y^t / z^{\pi(t)}\right)\\ &=&\CC\left(\left\{\bigwedge\limits_{t=1}^{q}y^t / z^{\pi(t)}~~|~~\pi\in\CA_q^m\right\}\right) \end{eqnarray*} The second statement results from the first statement, Lemma \ref{prop:algebra} and Remark \ref{rk:app:cone}. That $R[\WST[\CE],\CF]=\CC(\text{MO-CR})$ and then $\EFF[R[\WST[\CE],\CF]]=\text{MO-CR}$ (from Remark 1) concludes the proof. \end{proof} Ultimately, this proves the \textbf{correctness} of Algorithm 1 for the computation of MO-CR, given $\WST[\CE]=\{y^1,\ldots,y^q\}$ and $\CF=\{z^1,\ldots,z^m\}$. It consists in the iterative expansion/construction of the intersection $R[\WST[\CE],\CF]$, which can be seen as dynamic programming on the paths of the layer graph. 
For $k\in\{1,\ldots,q\}$, we denote $D^t$ the description of the cone-union corresponding to the intersection:  $$\CC(D^t)=\cap_{l=1}^{t} \cup_{z\in\CF} \CC(y^l/z)$$ Recursively, for $t>1$,  $\CC(D^t)=\CC(D^{t-1})~\cap~(\cup_{z\in\CF}~\CC(y^{t}/z))$. From Lemma \ref{prop:algebra},  Remark \ref{rk:app:cone} and Corollary 1, in order to construct, we  then have to iterate the following: $$ D^t=\EFF[\{\rho ~\wedge~ (y^t/z)~~|~~\rho\in D^{t-1},~~z\in\CF\}] $$
We now proceed with the \textbf{time complexity} of Algorithm 1. At first glance, since there are $m^q$ paths in the layer graph, then there are $O(m^q)$ elements in MO-CR. Fortunately, they are much less, because we have: \begin{theorem}[MO-CR is polynomially-sized]~\\ \label{th:mopoa:poly} Given a MOG and denoting $d=|\OBJ|$, $q=|\WST[\CE]|$ and $m=|\CF|$, we have: $$|\text{MO-CR}|\leq (qm)^{d-1}$$ \end{theorem} \begin{proof} Given $\rho\in\text{MO-CR}$, for some $\pi\in\CA_q^m$, we have $\rho=\bigwedge\nolimits_{t=1}^{q}y^t / z^{\pi(t)}$, and then $\forall k\in\OBJ, \rho_k=\min_{t=1\ldots q}\{y^t_k / z^{\pi(t)}_k\}$. Therefore, $\rho_k$ is exactly realized by the $k$th component of at least one cone apex $y^t / z^{\pi(t)}$ in the layer graph (that is a vertex in the layer-graph above). Consequently, there are at most as many possible values for the $k$th component of $\rho$, as the number of vertices in the layer graph, that is $qm$. This holds for the $d$ components of $\rho$; hence there are at most $(qm)^d$ vectors in MO-CR. More precisely, by Lemma \ref{lem:eff} (below), since MO-CR is an efficient set, then there are at most $(qm)^{d-1}$ vectors in MO-CR. \end{proof} \begin{lemma}\label{lem:eff} Let $Y\subseteq\REp$ be a set of vectors, with at most $M$  values on each component: $$|~\EFF[Y]~|\leq M^{d-1}$$ \end{lemma} \begin{proof}  For instance, in $\RE^2_+$, considering the $M\times M$ grid in the plane, there is at most one Pareto-efficient vector per column, hence $|\EFF[Y]|\leq M$. Think of each vector as having one  and $d-1$ components. Fixing these last components, a single-objective optimization problem on the first objective occurs. Hence there is one optimum. Furthermore, there are at most $M^{d-1}$ valuations realized on the $d-1$ other components. If you fix the $d-1$ last components, there is at most one Pareto-efficient vector: it maximizes the first component. \end{proof}In Algorithm 1, there are $\Theta(q)$ steps. At each step $t$, from Theorem \ref{th:mopoa:poly}, we know that $|D^{t-1}|\leq (qm)^{d-1}$. Hence, $|\{\rho ~\wedge~ (y^t/z)~~|~~\rho\in D^{t-1},~~z\in\CF\}|\leq q^{d-1} m^d$, and the computation of the efficient set $D^t$ requires time $O((q^{d-1} m^d)^2 d)$. 
Ultimately, Algorithm 1 takes $q$ steps and then time $O(q (q^{d-1} m^d)(qm)^{d-1} d)=O((qm)^{2d-1}d)$. If $d=2$, this lowers to $O((q m)^2\log_2(q m))$, by using a data structure (e.g. an AVL tree) that orders vectors according to the first objective and does comparisons on the second objective.

\section{Proof of Lemma 1}

This proof simply consists in chaining the quantifiers in the definitions, that have been carefully chosen to prove the result.

(1) First, let us  show  $R[E,F]\subseteq R[\WST[\CE],\CF]$. Let $\rho'$ be a ratio of $R[E,F]$ and let us show that: $$ \forall y\in \WST[\CE],~~ \exists z\in \CF,~~ \text{ s.t.: } y\WPP\rho'\star z $$ Take $y\in\WST[\CE]$. From the first condition, there is a $y'\in E$ such that $y\WPP y'$. From MO-CR, there is a $z'$ such that $y'\WPP \rho'\star z'$. From the third condition on $z'$, there exists $z\in\CF$ such that $z'\WPP z$. Recap: $y\WPP y'\WPP \rho'\star z'\WPP \rho'\star z$.

(2) Then, let $\rho$ be a ratio of $R[\WST[\CE],\CF]$, and let us  show that\linebreak $\rho'=(1+\varepsilon_1)^{-1}(1+\varepsilon_2)^{-1}\rho$~~ is in $R[E,F]$, that is: $$ \forall y'\in E,~~ \exists z'\in F,~~ (1+\varepsilon_1) y'\WPP (1+\varepsilon_2)^{-1} \rho\star z' $$ Take an element $y'$ of $E$. From the second condition, there is $y\in\WST[\CE]$ such that $(1+\varepsilon_1)y'\WPP y$.  From MO-CR, there is $z\in\CF$ such that $y\WPP\rho\star z$. From the fourth condition on $z$, there exists $z'\in F$ s.t. $z\WPP (1+\varepsilon_2)^{-1} z'$. Recap: $(1+\varepsilon_1)y'\WPP y\WPP \rho\star z\WPP (1+\varepsilon_2)^{-1} \rho\star z'$.

\section{Proof of Theorem 6}

Applying Algorithm 1 on $E$ and $F$ outputs an $((1+\varepsilon_1)(1+\varepsilon_2))$-covering of $R(\WST[\CE],\CF)$. Moreover, since we have $|E|=O((1/\varepsilon_1)^{d-1})$ and \linebreak $|F|=O((1/\varepsilon_2)^{d-1})$, Algorithm 1 takes time $O\left(d/(\varepsilon_1\varepsilon_2)^{(d-1)(2d-1)}\right)$.

\section{Experiments}

Experiments were conducted to assess the practicality of our polynomial time and approximation algorithms. We used C++STL on a Linux laptop equipped with CPUs at 1.40Ghz. We fixed $|A^i|=2$ actions per agent. For each parameter-values, we averaged the measures over 5 random instances\footnote{Though only 5 random instances  does not sound like much, the measures of cpu-time were already stable}. The evaluations $u^{i}_k(a^{\mathcal{N}(i)})$ are drawn uniformly and independently in $|[1,16]|$. In Table 1 (for MOGs) we have $\mathcal{N}(i)=N$. In Table 2 (for MO graphical games), the games were drawn on grid graphs with dimensions $n=n_1\times n_2$, in order to experiment various treewidths\footnote{For a formal definition of the treewidth, the reader may refer to \cite{dechter1989tree,jensen1994influence} or \cite{ismaili2016computational}.}. We chose $n_2\in\{1,2,3\}$ for the interaction-graph's width, which corresponds to the treewidths $\mathcal{T}\in\{2,4,6\}$. 

\subsection{Computational measures on MO normal forms}
In Table \ref{tab:mog}, we experiment Algorithm 1 on MO games. Table 1's notations are: $d$ for the number of objectives;  $n$ for the number of agents;  T(P1) for the cpu-time (seconds) of Phase 1: computing $\WST[\CE]$ and $\CF$; $m=|\CF|$ and $q=|\WST[\CE]|$;  T(P2) for the cpu-time (seconds) of Phase 2: computing MO-CR given $\WST[\CE]$ and $\CF$;  and finally, (in order to assess the practicality of the algorithm's output) the size of the resulting MO-CR.
\begin{table}[h] 
\centering 
$\begin{array}{c|ccc|ccc|ccc}																			 	&			\multicolumn{3}{c|}{d=2}			&			\multicolumn{3}{c|}{d=3}			&			\multicolumn{3}{c}{d=4}			\\ \hline																			 n	&	4	&	8	&	12	&	4	&	8	&	12	&	4	&	8	&	12	\\ \hline																			 T(P1)	&	0.00	&	0.08	&	2.40	&	0.00	&	0.07	&	2.45	&	0.00	&	0.08	&	2.44	\\ m	&	4.2	&	5.6	&	8.2	&	5.4	&	17.8	&	41.2	&	7.2	&	36.4	&	105.8	\\ q	&	2.2	&	4	&	5.8	&	4.4	&	9.8	&	30.2	&	8	&	37.8	&	82.6	\\ \hline T(P2)	&	0.00	&	0.00	&	0.00	&	0.00	&	0.01	&	0.20	&	0.00	&	0.48	&	30.85	\\ \#\text{MO-CR}	&	3.4	&	4.6	&	6	&	3.2	&	22.8	&	31.8	&	13.6	&	44.4	&	154.8	\\ \hline																			 
\end{array} $\\[1ex]
\caption{Computation times for Phases 1 and 2 on MO normal forms} \label{tab:mog} 
\end{table}	

\indent\textit{Observations.} Recall that the normal form is a representation of size $\Theta(n\alpha^n d)$. For instance, for $d=3$, Phase 1, and $n=4,8,12$, the instance to read is made of $192$, $6144$ and $147456$ scalars. The cpu-time cost of Phase 1 depends directly on the size of this input. For $d\leq 3$, Algorithm 1 costs nothing, compared to Phase 1. For $d\geq 4$, we begin to perceive the explosion of Algorithm 1 (Phase 2), while $m,q\simeq 100$. This indicates a practical intractability for $d\geq 4$. Recall that the cost of Algorithm 1 (Phase 2) for $d=4$ is $O((mq)^7)$.

\subsection{Computational measures on MO graphical games}
In Table 2, we experiment the approximation scheme  on MO graphical games. After Phase 1, we take smaller representations of $\WST[\CE]$ and $\CF$:  a $(1+\varepsilon_1)$-under-covering of $\WST[\CE]$ with $\varepsilon_1=6.5\%$,  and a $(1+\varepsilon_2)$-stick-covering of $\CF$ with $\varepsilon_2=3.5\%$, all in order to ensure a $(1+\varepsilon)$-covering of MO-CR, with $\varepsilon\simeq 6.5\%+3.5\%=10\%$ (thanks to Theorem 6). Table 2's notations are the same as Table 1's, and we add:  $n_2$ for the width of the interaction graph; $m_{\varepsilon}$ for the resulting size (after a proper rounding) of the representation of $\CF$; and $q_{\varepsilon}$ for the resulting size (after a proper rounding) of the representation of $\WST[\CE]$.

\begin{table}[h] 
\centering 
$\begin{array}{c|ccc|ccc|ccc}																			 	&	\multicolumn{9}{c}{d=2}																	\\
 \hline																			 	&	\multicolumn{3}{c|}{n_2=1}					&	\multicolumn{3}{c|}{n_2=2}					&	\multicolumn{3}{c}{n_2=3}					\\
  \hline																			 n	&	60	&	120	&	180	&	60	&	120	&	180	&	60	&	120	&	180	\\
   \hline																			 T(P1)	&	0	&	3	&	13	&	1	&	11	&	43	&	4	&	37	&	159	\\ m	&	92	&	212	&	347	&	76	&	186	&	316	&	65	&	174	&	300	\\ q	&	47	&	120	&	217	&	49	&	134	&	222	&	46	&	134	&	228	\\ \hline																			 m_{\varepsilon}	&	5	&	5	&	4.6	&	5.6	&	5	&	4.4	&	4.6	&	4.8	&	5	\\ q_{\varepsilon}	&	7.4	&	7.4	&	7	&	7.6	&	7.4	&	7.2	&	7	&	7.4	&	7.6	\\ \hline																			 T(P2)	&	0.00	&	0.00	&	0.00	&	0.00	&	0.00	&	0.00	&	0.00	&	0.00	&	0.00	\\ \#\text{MO-CR}	&	3.8	&	2.6	&	3.6	&	3.4	&	3.4	&	3	&	3.6	&	2.4	&	3	\\ \hline																			 \multicolumn{10}{c}{}\\																			 	&	\multicolumn{9}{c}{d=3}																	\\ \hline																			 	&	\multicolumn{3}{c|}{n_2=1}					&	\multicolumn{3}{c|}{n_2=2}					&	\multicolumn{3}{c}{n_2=3}					\\ \hline																			 n	&	12	&	24	&	36	&	12	&	24	&	36	&	12	&	24	&	36	\\ \hline																			 T(P1)	&	0	&	1	&	8	&	0	&	2	&	15	&	0	&	4	&	40	\\ m	&	42	&	263	&	777	&	44	&	249	&	596	&	49	&	190	&	474	\\ q	&	27	&	236	&	645	&	31	&	143	&	506	&	38	&	159	&	448	\\ \hline																			 m_{\varepsilon}	&	15.6	&	19.8	&	22.8	&	16.8	&	27.4	&	26.6	&	17.8	&	22	&	25.2	\\ q_{\varepsilon}	&	24.8	&	45.8	&	53.2	&	26.2	&	59.2	&	63	&	27.8	&	45.2	&	53	\\ \hline																			 T(P2)	&	0.01	&	0.02	&	0.03	&	0.01	&	0.05	&	0.06	&	0.01	&	0.03	&	0.06	\\ \#\text{MO-CR}	&	15.6	&	13.4	&	12.8	&	8.8	&	11.8	&	14.4	&	12.4	&	15	&	12.8	\\ \hline																			 
\end{array}$ \\[1ex]
\caption{Computation times for Approximations on MOGGs} 
\label{tab:mogg} 
\end{table}		

\textit{Observations.} As seen in Table 1 when $m,q\simeq 100$, computing MO-CR would be experimentally intractable, if done directly on $\WST[\CE]$ and $\CF$. Fortunately, thanks to the approximation scheme, Algorithm 1 costs almost nothing on the smaller representations of $\WST[\CE]$ and $\CF$, compared to computation of Phase 1.

\subsection{A raw example of MO-CR}

\begin{figure}
\centering
\includegraphics[scale=0.6]{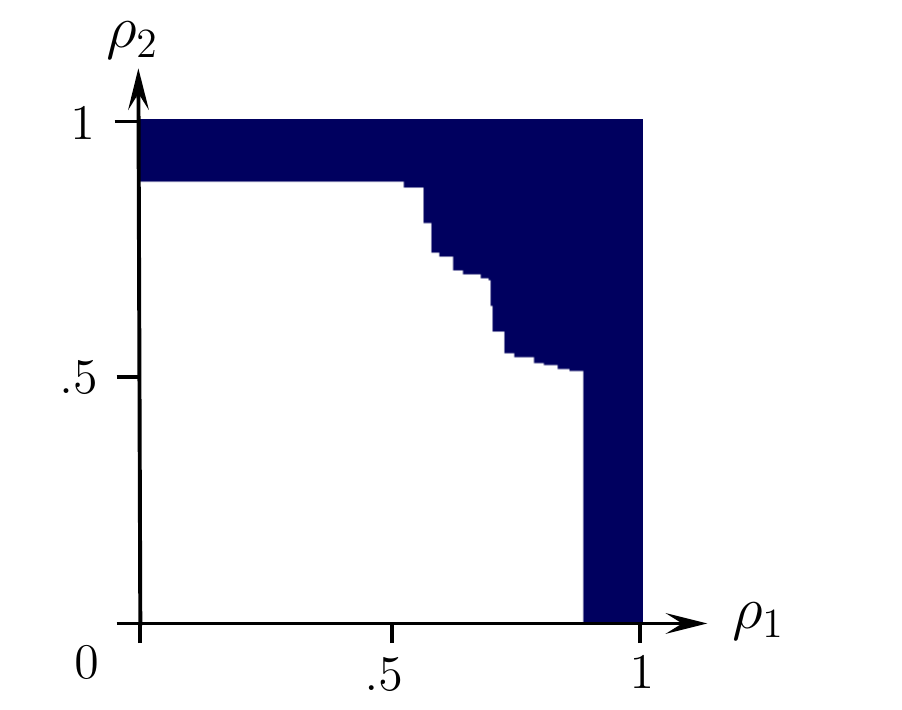}
\caption{MO-CR of one MOG with $n=7$ agents, $d=2$ objectives, $\alpha=3$ actions-per-agent, and independent evaluations drawn uniformly in $\{1,\ldots,100\}$}
\end{figure}
The white part corresponds to the set of guaranteed ratios of efficiency\\ $\rho\in R[\CE,\CF]\cap [0,1]^d$ 
and the dark-blue part to $\rho\notin R[\CE,\CF]$.
Recall that if $\rho\in R[\CE,\CF]$, then $\rho$ guarantees that for each equilibrium-outcome $y\in\CE$,  there exists an efficient-outcome $z^{(y)}$ such that $y \WPP \rho\star z^{(y)}$. Conversely, if $\rho\notin R[\CE,\CF]$, then there exists an \textit{in-efficient} equilibrium $y\in\CE$, that is: such that whatever $z\in\CF$, the guarantee $y \WPP \rho\star z^{(y)}$ does \textit{not} hold.
In other words, for each $\rho\in R[\CE,\CF]$, it holds that each equilibrium has at least $\rho$ times some efficiency.

\end{document}